\documentclass[12pt]{article}
\usepackage[dvips,letterpaper,margin=1.0in]{geometry}
\usepackage{sty}

\usepackage{authblk}
\author{Alexander S.\ Wein\thanks{Email: \textit{aswein@ucdavis.edu}. Part of this work was done while with the Algorithms and Randomness Center at Georgia Tech, supported by NSF grants CCF-2007443 and CCF-2106444.}}

\affil{Department of Mathematics, University of California, Davis}

\date{}

\title{Average-Case Complexity of Tensor Decomposition\\
for Low-Degree Polynomials}

\begin{document}

\maketitle

\begin{abstract}
Suppose we are given an $n$-dimensional order-3 symmetric tensor $T \in (\RR^n)^{\otimes 3}$ that is the sum of $r$ random rank-1 terms. The problem of recovering the rank-1 components is possible in principle when $r \lesssim n^2$ but polynomial-time algorithms are only known in the regime $r \ll n^{3/2}$. Similar ``statistical-computational gaps'' occur in many high-dimensional inference tasks, and in recent years there has been a flurry of work on explaining the apparent computational hardness in these problems by proving lower bounds against restricted (yet powerful) models of computation such as statistical queries (SQ), sum-of-squares (SoS), and low-degree polynomials (LDP). However, no such prior work exists for tensor decomposition, largely because its hardness does not appear to be explained by a ``planted versus null'' testing problem.

We consider a model for random order-3 tensor decomposition where one component is slightly larger in norm than the rest (to break symmetry), and the components are drawn uniformly from the hypercube. We resolve the computational complexity in the LDP model: $O(\log n)$-degree polynomial functions of the tensor entries can accurately estimate the largest component when $r \ll n^{3/2}$ but fail to do so when $r \gg n^{3/2}$. This provides rigorous evidence suggesting that the best known algorithms for tensor decomposition cannot be improved, at least by known approaches. A natural extension of the result holds for tensors of any fixed order $k \ge 3$, in which case the LDP threshold is $r \sim n^{k/2}$.
\end{abstract}

\newpage

\section{Introduction}

Tensor decomposition is a basic computational primitive underlying many algorithms for a variety of data analysis tasks, including phylogenetic reconstruction~\cite{phylo}, topic modeling~\cite{topic}, community detection~\cite{community,HS-bayesian,AAA,multilayer}, learning mixtures of Gaussians~\cite{mixture-1,mixture-2,smoothed,blessing}, independent component analysis~\cite{ica}, dictionary learning~\cite{BKS,MSS}, and multi-reference alignment~\cite{mra}. For further discussion, we refer the reader to~\cite{moitra-notes,survey-1,survey-2,survey-3}.

In this work we will consider order-$k$ tensors of the form $T \in (\RR^n)^{\otimes k}$ for an integer $k \ge 3$, that is, $T$ is an $n \times \cdots \times n$ ($k$ times) array of real numbers with entries denoted by $T_{i_1,\ldots,i_k}$ with $i_1,\ldots,i_k \in [n] := \{1,2,\ldots,n\}$. A vector $u = (u_i)_{i \in [n]} \in \RR^n$ can be used to form a rank-1 tensor $u^{\otimes k}$ defined by $(u^{\otimes k})_{i_1,\ldots,i_k} = u_{i_1} u_{i_2} \cdots u_{i_k}$, akin to the rank-1 matrix $uu^\top$.

In the \emph{tensor decomposition} problem, we observe a rank-$r$ tensor of the form
\begin{equation}\label{eq:rank-r-tensor}
T = \sum_{j=1}^r a_j^{\otimes k}
\end{equation}
with unknown components $a_j \in \RR^n$. The goal is to recover the components $a_1,\ldots,a_r$ up to the inherent symmetries in the problem, i.e., we cannot recover the ordering of the list $a_1,\ldots,a_r$ and if $k$ is even we cannot distinguish $a_j$ from $-a_j$. Our regime of interest will be $k$ fixed and $n$ large, with $r$ possibly growing with $n$.

A large number of algorithmic results for tensor decomposition now exist under various assumptions on $k,n,r$ and $\{a_j\}$~\cite{LRA,foobi,ica,smoothed,AGJ-1,AGJ-2,BKS,GM,fast-sos,MSS,GM-landscape,SS-fast,HSS-robust,KP-exact,fast-order-3}. We will focus on the order-3 case ($k=3$) to simplify the discussion. If $a_1,\ldots,a_r$ are linearly independent, the decomposition problem can be solved by a classical method based on simultaneous diagonalization\footnote{This method is sometimes called \emph{Jennrich's algorithm}, although this may be a misnomer; see~\cite{kolda-blog}.}~\cite{LRA}. Most recent work focuses on the more difficult \emph{overcomplete} regime where $r > n$. For \emph{random} order-3 tensors, where the components $a_j$ are drawn uniformly from the unit sphere, the state-of-the-art polynomial-time algorithms succeed when $r \ll n^{3/2}$ (where $\ll$ hides a polylogarithmic factor $(\log n)^{O(1)}$); this was achieved first by a method based on the sum-of-squares hierarchy~\cite{MSS} and then later by a faster spectral method~\cite{fast-order-3}. Under the weaker condition $r \lesssim n^2$ (where $\lesssim$ hides a constant factor), the decomposition is unique~\cite{BCO-id} and so the problem is solvable in principle. However, no polynomial-time algorithm is known in the regime $r \gtrsim n^{3/2}$.

For random tensors of any fixed order $k \ge 4$, the situation is similar. The decomposition is unique when $r \lesssim n^{k-1}$~\cite{BCO-id}, but poly-time algorithms are only known for substantially smaller rank. For $k = 4$, poly-time decomposition is known when $r \ll n^2$~\cite{foobi,MSS}, and it is expected that $r \ll n^{k/2}$ is achievable for general $k$; however, to our knowledge, the best poly-time algorithms in the literature for $k \ge 5$ require $r \lesssim n^{\lfloor (k-1)/2 \rfloor}$~\cite{LRA,smoothed}. These algorithms for $k \ge 4$ work not just for random tensors but also under much weaker assumptions on the components.

\paragraph{Statistical-computational gaps.}

The motivation for this work is to understand whether the ``statistical-computational gap'' mentioned above is inherent, that is, we hope to investigate whether or not a poly-time algorithm exists in the apparent ``possible but hard'' regime $n^{k/2} \lesssim r \lesssim n^{k-1}$. For context, gaps between statistical and computational thresholds are a nearly ubiquitous phenomenon throughout high-dimensional statistics, appearing in settings such as planted clique, sparse PCA, community detection, tensor PCA, and more; see e.g.~\cite{WX-survey,phys-survey} for exposition. In these average-case problems where the input is random, classical computational complexity theory has little to say: results on NP-hardness (including those for tensor decomposition~\cite{hastad-rank,tensor-hard}) typically show hardness for worst-case instances, which does not imply hardness in the average case. Still, a number of frameworks have emerged to justify why the ``hard'' regime is actually hard, and to predict the location of the hard regime in new problems.
One approach is based on average-case reductions which formally transfer suspected hardness of one average-case problem (usually planted clique) to another (e.g.~\cite{BR-reduction,HWX-reduction,BBH-reduction,secret-leakage}). Another approach is to prove unconditional lower bounds against particular algorithms or classes of algorithms (e.g.~\cite{jerrum,decelle,MMSE}), and sometimes this analysis is based on intricate geometric properties of the solution space (e.g.~\cite{alg-barriers,GS-ogp,GZ-ogp}; see~\cite{ogp-survey} for a survey). Perhaps some of the most powerful and well-established classes of algorithms to prove lower bounds against are statistical query (SQ) algorithms (e.g.~\cite{sq-clique,DKS-sq}), the sum-of-squares (SoS) hierarchy (e.g.~\cite{sos-clique,KMOW}; see~\cite{sos-survey} for a survey), and low-degree polynomials (LDP) (e.g.~\cite{HS-bayesian,sos-hidden,hopkins-thesis}; see~\cite{ld-notes} for a survey). In recent years, all of the above frameworks have been widely successful at providing many different perspectives on what makes statistical problems easy versus hard, and there has also been progress on understanding formal connections between different frameworks~\cite{sos-hidden,GJW-ld,sq-ld,fp}. For many conjectured statistical-computational gaps, we now have sharp lower bounds in one (or more) of these frameworks, suggesting that the best known algorithms cannot be improved (at least by certain known approaches).

Despite all these successes, the tensor decomposition problem is a rare example where (prior to this work) essentially no progress has been made on any kind of average-case hardness, even though the problem itself and the suspected threshold at $r \sim n^{k/2}$ have maintained a high profile in the community since the work of~\cite{GM} in 2015. The lecture notes~\cite{kunisky-lecture} highlight ``hardness of tensor decomposition'' as one of 15 prominent open problems related to sum-of-squares (Open Problem~7.1). In Section~\ref{sec:diff} we explain why tensor decomposition is more difficult to analyze than the other statistical problems that were understood previously, and in Section~\ref{sec:technique-lower} we explain how to overcome these challenges.

\paragraph{The LDP framework.}

Our main result (Theorem~\ref{thm:main}) determines the computational complexity of tensor decomposition in the \emph{low-degree polynomial (LDP) framework}, establishing an easy-hard phase transition at the threshold $r \sim n^{k/2}$. This means we consider a restricted class of algorithms, namely those that can be expressed as $O(\log n)$-degree multivariate polynomials in the tensor entries. The study of this class of algorithms in the context of high-dimensional statistics first emerged from a line of work on the sum-of-squares (SoS) hierarchy~\cite{sos-clique,HS-bayesian,sos-hidden,hopkins-thesis} and is now a popular tool for predicting and explaining statistical-computational gaps; see~\cite{ld-notes} for a survey. Low-degree polynomials capture various algorithmic paradigms such as spectral methods and approximate message passing (see e.g.\ Section~4 of~\cite{ld-notes}, Appendix~A of~\cite{GJW-ld}, and Theorem~1.4 of~\cite{MW-amp}), and so LDP lower bounds allow us to rule out a large class of known approaches all at once. For the types of problems that typically arise in high-dimensional statistics (informally speaking), the LDP framework has a great track record for consistently matching the performance of the best known poly-time algorithms. As a result, LDP lower bounds can be taken as evidence for inherent hardness of certain types of average-case problems. While there are some settings where LDPs are outperformed by another algorithm, these other algorithms tend to be ``brittle'' algebraic methods with almost no noise tolerance, and so the LDP framework is arguably still saying something meaningful in these settings; see~\cite{HW-counter,morris,KM-tree,lattice-1,lattice-2} for discussion.

Existing LDP lower bounds apply to a wide variety of statistical tasks, which can be classified as hypothesis testing (e.g.~\cite{HS-bayesian,sos-hidden,hopkins-thesis}), estimation (e.g.~\cite{SW-estimation,KM-tree}), and optimization (e.g.~\cite{GJW-ld,W-opt,BH-ksat}). The techniques used to prove these lower bounds are quite different in the three cases, and the current work introduces a new technique for the case of estimation.

\subsection{Largest Component Recovery}

In order to formulate the random tensor decomposition problem in the LDP framework, we will define a variant of the problem where one component has slightly larger norm than the rest, and the goal is to recover this particular component.

\begin{problem}[Largest component recovery]
\label{prob:largest}
Let $k \ge 3$ be odd and $\delta > 0$. In the \emph{largest component recovery problem}, we observe
\[ T = (1+\delta) a_1^{\otimes k} + \sum_{j=2}^r a_j^{\otimes k} \]
where each $a_j$ (for $1 \le j \le r$) is drawn uniformly and independently from the hypercube $\{\pm 1\}^n$. The goal is to recover $a_1$ with high probability.
\end{problem}

\noindent One can think of $\delta$ as a small constant, although our results will apply more generally. The purpose of increasing the norm of the first component is to give the algorithm a concrete goal, namely to recover $a_1$; without this, the algorithm has no way to break symmetry among the components. We have restricted to odd $k$ here because otherwise there is no hope to disambiguate between $a_1$ and $-a_1$; however, we give similar results for even $k$ in Section~\ref{sec:general}. The assumption that $a_j$ is drawn from the hypercube helps to make some of the proofs cleaner, but our approach can likely handle other natural distributions for $a_j$ such as $\mathcal{N}(0,I_n)$ (and we do not expect this to change the threshold).

\paragraph{Connection to decomposition.}

While Problem~\ref{prob:largest} does not have formal implications for the canonical random tensor decomposition model (i.e.,~\eqref{eq:rank-r-tensor} with $a_j$ uniform on the sphere), it does have implications for a mild generalization of this problem.

\begin{problem}[Semirandom tensor decomposition]
\label{prob:semirandom}
Let $k \ge 3$ be odd and $\delta > 0$. In the \emph{semirandom tensor decomposition problem}, we observe
\[ T = \sum_{j=1}^r \lambda_j a_j^{\otimes k} \]
where $\lambda_j \in [1,1+\delta]$ are arbitrary and $a_j$ are drawn uniformly and independently from the hypercube $\{\pm 1\}^n$. The goal is to recover $a_1,\ldots,a_r$ (up to re-ordering) as well as the corresponding $\lambda_j$'s, with high probability.
\end{problem}

\noindent Note that any algorithm for Problem~\ref{prob:semirandom} can be used to solve Problem~\ref{prob:largest} (with the same parameters $k,n,r,\delta$): simply run the algorithm to recover all the $(\lambda_j,a_j)$ pairs and then output the $a_j$ with the largest corresponding $\lambda_j$. When $r \gg n^{k/2}$, our main result (Theorem~\ref{thm:main}) will give rigorous evidence for hardness of Problem~\ref{prob:semirandom} via a two-step argument: we will show LDP-hardness of Problem~\ref{prob:largest}, justifying the conjecture that there is no poly-time algorithm for Problem~\ref{prob:largest}; this conjecture, if true, formally implies that there is no poly-time algorithm for Problem~\ref{prob:semirandom}.

We see Problem~\ref{prob:semirandom} as a natural variant of random tensor decomposition. In particular, we expect that existing algorithms~\cite{MSS,fast-order-3} can be readily adapted to solve this variant when $r \ll n^{k/2}$, at least for $k \in \{3,4\}$. As such, understanding the complexity of Problem~\ref{prob:semirandom} (via the connection to Problem~\ref{prob:largest} described above) can be viewed as the end-goal of this work. From this point onward, we will study Problem~\ref{prob:largest}. In Section~\ref{sec:future} we discuss possible extensions of our result that would more directly address the canonical model for random tensor decomposition.

\begin{remark}
In line with the discussion in the Introduction, our results should really only be considered evidence for hardness of Problems~\ref{prob:largest} and~\ref{prob:semirandom} in the presence of noise. Indeed, if $\delta$ is not an integer, there is a trivial algorithm for Problem~\ref{prob:largest} by examining the non-integer part of $T$ (the author thanks Bruce Hajek for pointing this out). This algorithm is easily defeated by adding Gaussian noise of variance $\sigma^2 \gg 1$ to each entry of $T$. This is a low level of noise, as we believe our low-degree upper bound tolerates $\sigma \ll n^{k/4}$, akin to tensor PCA.
\end{remark}

\subsection{Main Results}

We now show how to cast Problem~\ref{prob:largest} in the LDP framework and state our main results. The class of algorithms we consider are (multivariate) polynomials $f$ of degree (at most) $D$ with real coefficients, whose input variables are the $n^k$ entries of $T \in (\RR^n)^{\otimes k}$ and whose output is a real scalar value; we write $\RR[T]_{\le D}$ for the space of such polynomials. We will be interested in whether such a polynomial can accurately estimate $a_{11} := (a_1)_1$, the first entry of the vector $a_1$. Following~\cite{SW-estimation}, define the \emph{degree-$D$ minimum mean squared error}
\[ \MMSE_{\le D} := \inf_{f \in \RR[T]_{\le D}} \EE[(f(T) - a_{11})^2] \]
where the expectation is over $\{a_i\}$ and $T$ sampled as in Problem~\ref{prob:largest}. Note that the above ``scalar MMSE'' is directly related to the ``vector MMSE'': by linearity of expectation and symmetry among the coordinates,
\begin{equation}\label{eq:vector-mmse}
\inf_{f_1,\ldots,f_n \in \RR[T]_{\le D}} \EE \sum_{i=1}^n (f_i(T) - (a_1)_i)^2 = n \cdot \MMSE_{\le D}.
\end{equation}
It therefore suffices to study the scalar MMSE, with $\MMSE_{\le D} \to 0$ indicating near-perfect estimation and $\MMSE_{\le D} \to 1$ indicating near-trivial estimation (no better than the constant function $f(T) \equiv 0$).

\begin{theorem}\label{thm:main}
Fix any $k \ge 3$ odd and $\delta > 0$. As in Problem~\ref{prob:largest}, let
\[ T = (1+\delta) a_1^{\otimes k} + \sum_{j=2}^r a_j^{\otimes k} \]
where each $a_j$ (for $1 \le j \le r$) is drawn uniformly and independently from the hypercube $\{\pm 1\}^n$. Let $r = r_n$ grow with $n$ as $r = \Theta(n^{\alpha})$ for an arbitrary fixed $\alpha > 0$.
\begin{itemize}
\item (``Easy'') If $\alpha < k/2$ then $\lim_{n \to \infty} \MMSE_{\le C \log n} = 0$ for a constant $C = C(k,\delta) > 0$.
\item (``Hard'') If $\alpha > k/2$ then $\lim_{n \to \infty} \MMSE_{\le n^c} = 1$ for a constant $c = c(k,\alpha) > 0$.
\end{itemize}
\end{theorem}

\noindent More precise and more general statements of the lower bound (``hard'' regime) and upper bound (``easy'' regime) can be found in Section~\ref{sec:general} (Theorems~\ref{thm:main-lower} and~\ref{thm:main-upper}), where the case of $k$ even is also handled. The lower bound shows failure of LDP algorithms when $r \gg n^{k/2}$. Since no one has managed to find a poly-time algorithm in the LDP-hard regime for many other problems of a similar nature (including planted clique~\cite{hopkins-thesis}, community detection~\cite{HS-bayesian}, tensor PCA~\cite{sos-hidden,ld-notes}, and more), this suggests that there is no poly-time algorithm for Problem~\ref{prob:largest} (and therefore also for Problem~\ref{prob:semirandom}) when $r \gg n^{k/2}$, at least without a new algorithmic breakthrough. This is the first average-case hardness result of any type for tensor decomposition, and it gives a concrete reason to believe that the existing algorithms for $k=3$ (which succeed when $r \ll n^{3/2}$~\cite{MSS,fast-order-3}) are essentially optimal. Following the heuristic correspondence between polynomial degree and runtime in Hypothesis~2.1.5 of~\cite{hopkins-thesis} (see also~\cite{ld-notes,subexp-sparse})---namely, degree $D$ corresponds to runtime $\exp(\tilde\Theta(D))$ where $\tilde\Theta$ hides factors of $\log n$---the lower bound suggests that runtime $\exp(n^{\Omega(1)})$ is required in the hard regime.

The upper bound shows that a degree-$O(\log n)$ polynomial successfully estimates $a_1$ in the regime $r \ll n^{k/2}$. Such a polynomial has $n^{O(\log n)}$ terms and can thus be evaluated in time $n^{O(\log n)}$ (quasi-polynomial). This is in stark contrast to the $\exp(n^{\Omega(1)})$ runtime that we expect to need in the hard regime. We are confident that for $k \in \{3,4\}$ and $r \ll n^{k/2}$, Problems~\ref{prob:largest} and~\ref{prob:semirandom} can in fact be solved in \emph{polynomial} time by adapting existing algorithms~\cite{MSS,fast-order-3}, but we have not attempted this here as our main focus is on establishing the phase transition for polynomials. We consider the lower bound to be our main contribution, and the upper bound serves to make the lower bound more meaningful.

\subsection{Discussion}

\subsubsection{Why the LDP framework?}

Since there are many different frameworks for exploring statistical-computational gaps, we briefly discuss the possibility of applying some of the others to tensor decomposition, and explain why the LDP framework seems especially suited to this task.

\paragraph{Statistical query (SQ) model.}

The SQ model (e.g.~\cite{sq-clique,DKS-sq}) is applicable only to settings where the observation consists of i.i.d.\ samples from some unknown distribution. There does not seem to be a clear way to cast tensor decomposition in this form.

\paragraph{Sum-of-squares (SoS) hierarchy.}

SoS is a powerful algorithmic tool, and SoS lower bounds (e.g.~\cite{sos-clique,KMOW,sos-survey}) are sometimes viewed as a gold standard in average-case hardness. However, it is important to note that SoS lower bounds show hardness of \emph{certification} problems, whereas Problems~\ref{prob:largest} and~\ref{prob:semirandom} are \emph{recovery} problems. Hardness of certification does not necessarily imply hardness of recovery, as discussed in~\cite{sk-cert,local-stats,quiet-coloring}. Therefore, while there are some natural certification problems related to tensor decomposition that would be nice to have SoS lower bounds for (e.g.\ Open Problem~7.2 in~\cite{kunisky-lecture}), it is not clear how to use an SoS lower bound to directly address the decomposition problem.

\paragraph{Optimization landscape.}

Prior work~\cite{AGJ-1,GM-landscape} has studied the landscape of a certain non-convex optimization problem related to tensor decomposition, leading to some algorithmic results. In other settings, recent work has studied structural properties of optimization problems as a means to prove failure of certain algorithms (e.g., local search, Markov chains, and gradient descent~\cite{GZ-ogp,BGJ-tensor,submatrix-ogp,ogp-sparse}) to recover a planted structure. If results of this type were obtained for tensor decomposition, they would complement ours by ruling out different types of algorithms. One caveat is that recent work on tensor PCA~\cite{RM-tensor-pca,BGJ-tensor,passed,kikuchi,iron-rough} and planted clique~\cite{GZ-clique,clique-elude} has revealed that the choice of which function to optimize can be very important. Thus, ``hardness'' of one particular canonical optimization landscape need not indicate inherent hardness of the recovery problem.

\paragraph{Reductions.}

Reductions from planted clique (e.g.~\cite{BR-reduction,HWX-reduction,BBH-reduction,secret-leakage}) are among the strongest forms of average-case hardness that exist for statistical problems, and it would be great to have reduction-based evidence for hardness of tensor decomposition. This would likely require new ideas beyond those in the current literature, as Problem~\ref{prob:largest} is somewhat different in structure from planted clique and the problems it has been reduced to thus far. These differences are explained in Section~\ref{sec:diff} below.

\bigskip

In summary, there are many complementary approaches that could be explored in future work in order to give new perspectives on hardness of tensor decomposition. In this work we have chosen to pursue the LDP framework, some strengths of which include its stellar track record of predicting statistical-computational gaps in the past, and its flexibility to directly address the decomposition problem (Problem~\ref{prob:semirandom}) rather than a related certification or optimization problem.

\subsubsection{Difficulties of tensor decomposition}
\label{sec:diff}

Although the suspected threshold $r \sim n^{k/2}$ for random tensor decomposition has been around for some time now~\cite{GM,MSS}, and despite the many recent successes in computational lower bounds for other statistical problems, there are (prior to this work) effectively no results that point to whether $r \sim n^{k/2}$ is a fundamental barrier or not. Here we will discuss what makes tensor decomposition more difficult to analyze than the other statistical-computational gaps that have received recent attention in the literature.

As a point of comparison, we briefly introduce the \emph{tensor principal component analysis (tensor PCA)} problem~\cite{RM-tensor-pca}: here we observe an order-$k$ tensor $Y = \lambda x^{\otimes k} + Z$ where $\lambda > 0$ is a signal-to-noise ratio, $x \in \RR^n$ is drawn uniformly from the unit sphere, and $Z \in (\RR^n)^{\otimes k}$ is a tensor with i.i.d.\ $\mathcal{N}(0,1)$ entries; the goal is to (approximately) recover the vector $x$ (up to sign, if $k$ is even). Tensor PCA also has a statistical-computational gap, and it is now well-understood: the best known poly-time algorithms require $\lambda \gtrsim n^{k/4}$, and there are matching lower bounds in both the SoS and LDP frameworks~\cite{tensor-pca-sos,sos-hidden,PR-sos-cert,ld-notes}. While tensor decomposition and tensor PCA may look superficially similar, tensor decomposition is much harder to analyze for the reason described below.

With a few exceptions (\cite{SW-estimation,KM-coloring,KM-tree}), essentially all existing lower bounds for recovering a planted signal in the SQ, SoS, or LDP frameworks leverage hardness of testing between a ``planted'' distribution and a ``null'' distribution, where the null distribution has independent entries. In the case of tensor PCA, in order to show hardness of recovering $x$ when $\lambda \ll n^{k/4}$, the SoS and LDP lower bounds crucially leverage the fact that it is hard to even \emph{detect} the planted structure when $\lambda \ll n^{k/4}$, that is, it is hard to distinguish between a spiked tensor $Y$ and an i.i.d.\ Gaussian tensor $Z$. Now for the decomposition problem, we are hoping to show hardness of decomposition whenever $r \gg n^{k/2}$, but the problem of distinguishing between a random rank-$r$ tensor and the appropriate Gaussian tensor\footnote{The Gaussian tensor matches moments of order $1$ and $2$ with the rank-$r$ tensor. (One could choose a different null distribution, but the analysis may be difficult; for instance, see Section~\ref{sec:future}.)} is actually easy when $r \ll n^k$; this can be achieved by thresholding a particular degree-4 polynomial in the tensor entries, where each monomial forms a double cover of $2k$ elements of $[n]$. This ``detection-recovery gap'' creates difficulty because the standard tools for proving lower bounds cannot show hardness of recovery unless detection is also hard. The tools with this limitation include the \emph{SDA (statistical dimension on average)} approach for SQ~\cite{sq-clique}, the \emph{pseudo-calibration} approach for SoS~\cite{sos-clique}, and the \emph{low-degree likelihood ratio}~\cite{HS-bayesian,sos-hidden,hopkins-thesis,ld-notes} (as well as its conditional variants~\cite{fp,grp-testing}) for LDP. Reductions from planted clique also typically require hardness of detection~\cite{BBH-reduction}.

The method of~\cite{SW-estimation} overcomes this challenge in some settings and gives LDP lower bounds for recovery problems even in regimes where detection is easy. This methods applies to problems of the form (roughly speaking) ``signal plus noise'' where the noise has independent entries. For instance, this gives LDP-hardness of recovery for variants of tensor PCA with detection-recovery gaps~\cite{LZ-tensor}. However, tensor decomposition does not take this ``signal plus noise'' form: Problem~\ref{prob:largest} has a ``signal'' term $(1+\delta)a_1^{\otimes k}$ but the remaining ``noise'' term is a rank-$(r-1)$ tensor, which does not have independent entries. As a result, we will need to introduce a new method in order to prove a lower bound for tensor decomposition. The proof strategy is outlined in Section~\ref{sec:technique-lower} below.

\subsection{Proof Techniques}

\subsubsection{Lower bound}
\label{sec:technique-lower}

Instead of mean squared error it will be convenient to work with an equivalent quantity, the \emph{degree-$D$ maximum correlation}
\[ \Corr_{\le D} := \sup_{f \in \RR[T]_{\le D}} \frac{\EE[f(T) \cdot a_{11}]}{\sqrt{\EE[f(T)^2]}}. \]
This is directly related to $\MMSE_{\le D}$ via the identity $\MMSE_{\le D} = \EE[a_{11}^2] - \Corr_{\le D}^2 = 1 - \Corr_{\le D}^2$ (see~\cite[Fact~1.1]{SW-estimation}), so our objective is to show $\Corr_{\le D} = o(1)$ when $r \gg n^{k/2}$ and $D = n^{\Omega(1)}$.

A first attempt is to compute $\Corr_{\le D}$ explicitly. Expand an arbitrary degree-$D$ polynomial in the monomial basis: $f(T) = \sum_{0 \le |S| \le D} \hat f_S T^S$ where $S$ ranges over multi-sets of tensor entries and $T^S := \prod_{I \in S} T_I$. The numerator of $\Corr_{\le D}$ is a linear functional of the coefficient vector $\hat f = (\hat f_S)_{|S| \le D}$:
\[ \EE[f(T) \cdot a_{11}] = \sum_S \hat f_S \, \EE[T^S \cdot a_{11}] = c^\top \hat f \qquad \text{where } c_S := \EE[T^S \cdot a_{11}]. \]
For the denominator, $\EE[f(T)^2]$ is a quadratic form:
\[ \EE[f(T)^2] = \sum_{S,S'} \hat f_S \hat f_{S'} \, \EE[T^S T^{S'}] = \hat f^\top P \hat f \qquad \text{where } P_{S,S'} := \EE[T^S T^{S'}]. \]
This means we have the explicit formula
\[ \Corr_{\le D} = \sup_{\hat f} \frac{c^\top \hat f}{\sqrt{\hat f^\top P \hat f}} = \sqrt{c^\top P^{-1} c}. \]
The difficulty here is that, while we can write down an explicit formula for the vector $c$ and matrix $P$, it does not seem tractable to write down an explicit formula for the \emph{inverse} matrix $P^{-1}$. We will instead find an upper bound on $\Corr_{\le D}$ that is manageable to work with.

Our difficulties stem from the fact that we do not have an orthogonal basis of polynomials to work with. After all, if $\{T^S\}_{|S| \le D}$ were orthogonal polynomials---i.e., if $P$ were a diagonal matrix---then we would have no problem inverting $P$. Since the distribution of $T$ is complicated (in particular, it is not a product measure), it seems difficult to construct an explicit basis of orthogonal polynomials. Instead, we will think of $f(T)$ as a function of the underlying i.i.d.\ Rademacher random variables $A = (a_{ij}) := (a_j)_i$ and work with an orthogonal basis of polynomials in those variables. For any $f(T)$ there is an associated polynomial $g(A)$ such that $f(T) = g(A)$, and we can expand these as
\[ \sum_{|S| \le D} \hat f_S T^S = f(T) = g(A) = \sum_{|U| \le kD} \hat g_U A^U, \]
where $U \subseteq [n] \times [r]$, $A^U := \prod_{(i,j) \in U} a_{ij}$, and $\hat g = (\hat g_U)_{|U| \le kD}$ is some vector of coefficients. Here $\{A^U\}$ is the standard Fourier basis for functions on the hypercube, which crucially does have the desired orthogonality property: $\EE[A^U A^{U'}] = \One_{U=U'}$. As a result, we can express the denominator of $\Corr_{\le D}$ as simply
\[ \EE[f(T)^2] = \EE[g(A)^2] = \|\hat g\|^2. \]
To exploit this, we will need to understand the relation between $\hat f$ and the corresponding $\hat g$. We will write down an explicit linear transformation, i.e., a matrix $M$ such that $\hat g = M \hat f$. The crux of our new strategy is that to obtain an upper bound on $\Corr_{\le D}$ it suffices to construct a left-inverse for $M$, i.e., a matrix $M^+$ such that $M^+ M = I$. To see why this suffices,
\[ \Corr_{\le D} = \sup_{\hat f} \frac{c^\top \hat f}{\|M \hat f\|} = \sup_{\hat f} \frac{c^\top M^+ M \hat f}{\|M \hat f\|} \le \sup_h \frac{c^\top M^+ h}{\|h\|} = \|c^\top M^+\|, \]
where in the inequality step, $h$ plays the role of $\hat g = M \hat f$ except we have relaxed the problem to allow $h$ to be \emph{any} vector, not necessarily in the image of $M$.

In light of the above, our remaining task is to construct an explicit left-inverse $M^+$. There are many possible choices, some of which will yield better bounds on $\Corr_{\le D}$ than others. We will need one that yields a good bound but is also analytically tractable to write down explicitly. We now explain the intuition behind our construction for $M^+$. Note that the left-inverse property is equivalent to $M^+ \hat g = \hat f$ whenever $\hat g = M \hat f$. That is, $M^+$ is a procedure to recover the (coefficients of the) polynomial $f$ when given the (coefficients of the) polynomial $g$ that satisfies $f(T) = g(A)$. Luckily there is a methodical process for this task which starts from the highest-degree terms and works downward. At each stage, there is always a particular monomial in $A$ whose coefficient in $g$ allows us to immediately deduce some particular coefficient in $f$. To illustrate, if $k = 3$ and $D = 2$ then the coefficient of the monomial
\begin{equation}\label{eq:example-monomial}
a_{27}a_{37}a_{47}a_{28}a_{38}a_{58}
\end{equation}
in $g(A)$ allows us to immediately deduce the coefficient of the monomial $T_{234} T_{235}$ in $f(T)$, since this is the only monomial in $T$ (of degree at most $2$) whose expansion in terms of $A$ contains the term~\eqref{eq:example-monomial}. Once we have deduced the highest-degree coefficients of $f$, we can subtract off the corresponding terms in $g$ and repeat the process recursively. Our choice of $M^+$ is inspired by this process; the full details can be found in Section~\ref{sec:construct-left}.

Recall now that the final step is to bound $\|c^\top M^+\|$ for our choice of $M^+$. Due to the recursive structure of $M^+$, this boils down to analyzing certain recursively-defined values $v_S$ (see Section~\ref{sec:rec-v}). As above, $S$ is a multi-set of tensor entries, which can be thought of as a hypergraph on vertex set $[n]$ (with a hyperedge for each tensor entry). The value $v_S$ is computed by summing over all possible ways to reduce $S$ to a smaller hypergraph by replacing some collection of hyperedges by a new hyperedge equal to their symmetric difference. Analyzing this recurrence is the most technically involved part of the proof, since there are subtle cancellations that need to be understood in order to get the right bound.

\subsubsection{Upper bound}

To construct a polynomial that accurately estimates the largest component, we take inspiration from a line of algorithmic work that builds spectral methods from tensor networks~\cite{fast-sos,tensor-net,fast-order-3,tensor-theory}. A tensor network is a diagram that specifies how to ``multiply'' together a collection of tensors to form another tensor; see~\cite{tensor-net} for details. Some existing algorithms for order-3 tensor decomposition~\cite{fast-sos,fast-order-3} (including the fastest known method that reaches the threshold $r \sim n^{3/2}$~\cite{fast-order-3}) are based on the idea of building a tensor network from multiple copies of the input tensor, flattening the result to a matrix, and running a spectral method on this matrix. Morally speaking, these are all steps that one should be able to capture using low-degree polynomials, but the algorithm of~\cite{fast-order-3} is actually a multi-stage process that seems challenging to write as a single polynomial. Also, since our goal is to estimate the largest component, we do not need to inject randomness to break symmetry among the components as in~\cite{fast-order-3}.

To construct our polynomial, we multiply $O(\log n)$ copies of $T$ together in a tensor network whose shape is an expander graph. The output is a single vector, which we take as our estimate for $a_1$. There is one key trick we use to greatly simplify the analysis: we deviate slightly from the standard notion of a tensor network by disallowing repeated ``edge labels.'' The precise definition of the polynomial appears in Section~\ref{sec:construct-poly}.

\subsection{Future Directions}
\label{sec:future}

We collect here a list of possible extensions of our main result and related open problems.

\begin{itemize}
    
    \item {\bf Distribution of $a_j$:} We have taken $a_j$ drawn uniformly from the hypercube in the interest of keeping the proofs as clean as possible. We expect our approach could be adapted to other distributions for $a_j$ such as $\mathcal{N}(0,I_n)$.
    
    \item {\bf Canonical model:} While we feel that the semirandom tensor decomposition model (Problem~\ref{prob:semirandom}) is quite natural, there may be interest in showing hardness for the more canonical model where every $\lambda_j$ is equal to 1. One way to approach this could be to study the task of hypothesis testing between a random rank-$r$ tensor and a random rank-$(r+1)$ tensor. Hardness (for poly-time algorithms) of this testing problem implies hardness of decomposition in the canonical model. It may be possible to use our proof technique to show LDP-hardness of this testing problem. Another alternative approach could be to consider a variant of Problem~\ref{prob:largest} where symmetry between the tensor components is broken by giving the algorithm a small amount of side-information rather than increasing the norm of one component.
    
    \item {\bf Tensors with algebraic structure:} Hopefully the proof techniques introduced here will open the door to other hardness results for related problems. For instance, \emph{orbit recovery problems}---a class of statistical problems involving group actions~\cite{orbit,orbit-2}---give rise to variants of tensor decomposition with algebraic structure. One example is \emph{multi-reference alignment (MRA)}, where the tensor components are cyclic shifts of one another~\cite{mra}. A statistical-computational gap in \emph{heterogeneous MRA} was conjectured in~\cite{het-mra}; the positive side was resolved in~\cite{wein-thesis} using~\cite{MSS}, and the negative side should now be approachable using our techniques. Tensors with continuous group symmetry are even less well understood~\cite{tensor-net,LM-so3}.
    
    \item {\bf Smoothed order-3 tensor decomposition:} For $k=3$, the algorithms achieving the optimal condition $r \ll n^{3/2}$ seem to crucially exploit randomness in the components~\cite{MSS,fast-order-3}. In contrast, other algorithms succeed in the \emph{smoothed analysis model} (which imposes minimal assumptions on the components) but require~$r \lesssim n$~\cite{ica,smoothed}. It remains unclear whether better algorithms for the smoothed model exist, or whether there is an inherent gap between the random and smoothed settings. For $k=4$ there is no such gap, with both random and smoothed algorithms achieving $r \sim n^2$~\cite{foobi,MSS}.
    
    \item {\bf Other LDP estimation lower bounds:} Proving LDP lower bounds for estimation problems remains a difficult task in many settings, with relatively few tools available compared to detection (hypothesis testing) problems. A few open problems are to resolve the LDP recovery threshold for sparse linear regression and group testing; these problems have detection-recovery gaps and the LDP \emph{detection} thresholds were resolved in~\cite{fp,grp-testing}.
    
\end{itemize}

\section{General Setting and Results}\label{sec:general}

It will be convenient to generalize Problem~\ref{prob:largest} in a few different ways. First, we allow all the components to have potentially different norms. This way, there is nothing distinguished about the first component, eliminating some tedious casework. To this end, we will consider tensors of the form
\begin{equation}\label{eq:obs-tensor}
\sum_{j=1}^r \lambda_j a_j^{\otimes k}
\end{equation}
where each component $a_j$ is drawn uniformly and independently from the hypercube $\{\pm 1\}^n$, and $\lambda_j \in \RR$ are deterministic nonzero scalars that are known to the algorithm. Without loss of generality we assume $1 = \lambda_1 \ge |\lambda_2| \ge |\lambda_3| \ge \cdots \ge |\lambda_r| =: \lambda_{\min} > 0$. The goal is to recover the largest component $a_1$. The case $\lambda_2 = \lambda_3 = \cdots = \lambda_r = (1+\delta)^{-1}$ is equivalent to Problem~\ref{prob:largest}.

Second, we will potentially allow the algorithm access not just to an order-$k$ tensor but also to tensors of lower order. For $I \subseteq [n]$ write
\begin{equation}\label{eq:T_I}
T_I := \sum_{j=1}^r \lambda_j a_j^I
\end{equation}
with $\lambda_j$ and $a_j$ as above, and where for a vector $v \in \RR^n$,
$v^I := \prod_{i \in I} v_i$. Note that since the $a_j$ have $\{\pm 1\}$-valued entries, each entry of the tensor~\eqref{eq:obs-tensor} takes the form $T_I$ for some $0 \le |I| \le k$.

Let $\Omega$ be a collection of subsets of $[n]$. We consider the task of estimating $a_{11} := (a_1)_1$ using a degree-$D$ polynomial $f: \RR^\Omega \to \RR$ whose input variables are $\{T_I\}_{I \in \Omega}$. Accordingly, we define
\[ \MMSE_{\le D}^\Omega := \inf_{\substack{f: \RR^\Omega \to \RR \\ \deg(f) \le D}} \EE[(f(T) - a_{11})^2]. \]

To make our \emph{lower bound} for order-$k$ tensors as strong as possible, we will choose $\Omega = \{I \subseteq [n] \,:\, 0 < |I| \le k\}$. This is equivalent to giving the algorithm access to all the tensors
\[ \sum_{j=1}^r \lambda_j a_j, \qquad \sum_{j=1}^r \lambda_j a_j^{\otimes 2}, \qquad \ldots, \qquad \sum_{j=1}^r \lambda_j a_j^{\otimes k}, \]
(which is a situation that often arises when using the method of moments, e.g.~\cite{mixture-1,mra}). Note that a lower bound in this setting implies a lower bound in the original setting where only the order-$k$ tensor is revealed.

To make our \emph{upper bound} as strong as possible, we will give the algorithm access to minimal information. For $k$ odd, we will take $\Omega = \{I \subseteq [n] \,:\, |I| = k\}$, that is, our polynomial only needs access to the ``off-diagonal'' entries of the order-$k$ tensor ($T_{i_1,\ldots,i_k}$ where $i_1,\ldots,i_k$ are all distinct). For $k$ even, the order-$k$ tensor alone does not suffice to disambiguate between $a_1$ and $-a_1$, so we additionally reveal the (off-diagonal part of the) order-$(k-1)$ tensor, that is, we take $\Omega = \{I \subseteq [n] \,:\, k-1 \le |I| \le k\}$. (An alternative formulation for $k$ even could be to reveal only the order-$k$ tensor and ask the algorithm to estimate $(a_1)_1 \cdot (a_1)_2$; we expect this could be analyzed using our methods.)

The following are non-asymptotic statements of our main results. Together, these immediately imply Theorem~\ref{thm:main}.

\begin{theorem}[Lower bound]\label{thm:main-lower}
Consider the setting described above with any parameters $k \ge 3$, $n \ge 1$, $D \ge 0$, and set $\Omega = \{I \subseteq [n] \,:\, 0 < |I| \le k\}$. If
\begin{equation}\label{eq:r-assum-hard}
r \ge 19 k^k D^{k+4} \lambda_{\min}^{-2} n^{k/2}
\end{equation}
then
\[ \MMSE_{\le D}^\Omega \ge 1 - n^{-1/2}. \]
In particular, if $k \ge 3$ is fixed, $\lambda_{\min} \ge \delta$ for fixed $\delta > 0$, and $r = r_n$ grows as $r = \Theta(n^\alpha)$ for fixed $\alpha > k/2$, then $\lim_{n \to \infty} \MMSE_{\le n^c}^\Omega = 1$ for a constant $c = c(k,\alpha) > 0$.
\end{theorem}

\begin{theorem}[Upper bound]\label{thm:main-upper}
Consider the setting described above with any $k \ge 3$. Let $k' \in \{k-1,k\}$ be odd, and set $\Omega = \{I \subseteq [n] \,:\, k' \le |I| \le k\}$. Suppose $n \ge n_0$ for some constant $n_0 = n_0(k)$. Let $D$ be the smallest odd integer such that
\[ D \ge \frac{k \log n}{1-|\lambda_2|}. \]
If $|\lambda_2| \le 1 - n^{-1/52}$ and
\begin{equation}\label{eq:r-assum-easy}
r \le \frac{1}{2} k^{-k/2} D^{-27k} n^{k/2}
\end{equation}
then
\begin{equation}\label{eq:mmse-upper-bound}
\MMSE_{\le D}^\Omega \le 10 k^{k-1} D^{52k} \frac{r}{n^{k-1-\One_{k\text{ even}}}} \le 10 k^{k-1} D^{52k} \frac{r}{n^{k/2}}.
\end{equation}
In particular, if $k \ge 3$ is fixed, $|\lambda_2| \le 1-\delta$ for fixed $\delta > 0$, and $r = r_n$ grows as $r = \Theta(n^\alpha)$ for fixed $0 < \alpha < k/2$, then $\lim_{n \to \infty} \MMSE_{\le C \log n}^\Omega = 0$ for a constant $C = C(k,\delta) > 0$.
\end{theorem}

\noindent Note that~\eqref{eq:r-assum-easy} is the bottleneck, requiring $r \ll n^{k/2}$. The intermediate step in~\eqref{eq:mmse-upper-bound} is a sharper bound on the MMSE that we use for the remark below, but it does not dictate the threshold for $r$.

\begin{remark}[Exact recovery]
We remark that if $\MMSE_{\le D} = o(1/n)$ then \emph{exact} recovery of $a_1$ is possible with probability $1-o(1)$ by thresholding the values of certain degree-$D$ polynomials $f_1,\ldots,f_n$. To see this: combining~\eqref{eq:vector-mmse} with Markov's inequality, we have $f_1,\ldots,f_n$ such that $\sum_{i=1}^n (f_i(T) - (a_1)_i)^2 < 1$ with probability $1 - o(1)$. Furthermore, thresholding $f_1,\ldots,f_n$ is guaranteed to exactly recover $a_1$ whenever $\sum_{i=1}^n (f_i(T) - (a_1)_i)^2 < 1$.

As a result, when $|\lambda_2| \le 1 - \Omega(1)$ and $r \le n^{k/2} / (\log n)^{O(1)}$, our upper bound (Theorem~\ref{thm:main-upper}) gives exact recovery for any fixed $k \ge 5$. We expect a similar result to hold for $k \in \{3,4\}$ but this may require modifying the construction of the polynomial in the proof of Theorem~\ref{thm:main-upper}.
\end{remark}

\section{Proof of Theorem~\ref{thm:main-lower}: Lower Bound}

\subsection{Setup}

Throughout, define $\Omega := \{I \subseteq [n] \,:\, 0 < |I| \le k\}$. Our goal will be to give an upper bound on
\[ \Corr_{\le D} := \sup_{\substack{f: \RR^\Omega \to \RR \\ \deg(f) \le D}} \frac{\EE[f(T) \cdot a_{11}]}{\sqrt{\EE[f(T)^2]}}. \]
This will imply the desired result due to the following direct relation between $\Corr_{\le D}$ and the associated MMSE.

\begin{fact}[{\cite[Fact~1.1]{SW-estimation}}]
\label{fact:corr-mmse}
$\MMSE_{\le D}^\Omega = \EE[a_{11}^2] - \Corr_{\le D}^2 = 1 - \Corr_{\le D}^2$.
\end{fact}

Any degree-$D$ polynomial $f: \RR^\Omega \to \RR$ admits an expansion of the form
\[ f(T) = \sum_{0 \le |S| \le D} \hat f_S T^S \]
for some coefficients $\hat f_S \in \RR$, where $S$ takes values $S \in \NN^\Omega := \{0,1,2,\ldots\}^\Omega$ with $|S| := \sum_{I \in \Omega} S_I$ and $T^S := \prod_{I \in \Omega} T_I^{S_I}$.

At the same time, $T$ can be thought of as a function of $A$, the $n \times r$ matrix with columns $a_j$. This means we can also expand $f$ as
\[ f(T) = g(A) = \sum_U \hat g_U A^U \]
where $U$ ranges over subsets $U \subseteq [n] \times [r]$ of cardinality $|U| \le kD$. This expansion will be useful because, since $A$ has i.i.d.\ Rademacher entries, $\{A^U\}$ forms an orthonormal basis in the sense $\EE[A^U A^{U'}] = \One_{U = U'}$. As a consequence,
\begin{equation}\label{eq:g-norm}
\EE[f(T)^2] = \EE[g(A)^2] = \|\hat g\|^2 := \sum_U \hat g_U^2.
\end{equation}
Any vector of coefficients $\hat f = (\hat f_S)_{|S| \le D}$ induces a unique choice of $\hat g = (\hat g_U)_{|U| \le kD}$ such that
\begin{equation}\label{eq:f-g}
\sum_{|S| \le D} \hat f_S T^S = f(T) = g(A) = \sum_{|U| \le kD} \hat g_U A^U.
\end{equation}
It will be important to understand this mapping from $\hat f$ to $\hat g$.

\subsection{Proof Overview}

We will show that the mapping from $\hat f$ to $\hat g$ in~\eqref{eq:f-g} is a linear transformation that takes the form $\hat g = M \hat f$ for an explicit matrix $M = (M_{US})$. A key step in the proof will be to construct an explicit \emph{left inverse} for $M$, that is, a matrix $M^+$ satisfying $M^+ M = I$. In other words, $M^+$ is a matrix that recovers the coefficients $\hat f$ from the coefficients $\hat g$: $M^+ \hat g = M^+ M \hat f = \hat f$.

The numerator of $\Corr_{\le D}$ can be expressed as
\[ \EE[f \cdot a_{11}] = \sum_{|S| \le D} \hat f_S \, \EE[T^S \cdot a_{11}] = c^\top \hat f
\]
where the vector $c = (c_S)_{|S| \le D}$ is defined by
\begin{equation}\label{eq:c-def}
c_S = \EE[T^S \cdot a_{11}].
\end{equation}
Using~\eqref{eq:g-norm}, the denominator can be expressed as $\sqrt{\EE[f(T)^2]} = \|\hat g\| = \|M \hat f\|$. This means we can write
\begin{equation}\label{eq:corr-bound}
\Corr_{\le D} = \sup_{\hat f} \frac{c^\top \hat f}{\|M \hat f\|} = \sup_{\hat f} \frac{c^\top M^+ M \hat f}{\|M \hat f\|} \le \sup_h \frac{c^\top M^+ h}{\|h\|} = \|c^\top M^+\|.
\end{equation}
In the crucial \emph{inequality} above, $h$ plays the role of $\hat g = M \hat f$ except we have relaxed the problem to allow $h$ to be \emph{any} vector, not necessarily in the image of $M$. After this simplification, the optimizer for $h$ is $h^* = (c^\top M^+)^\top$ (or any scalar multiple thereof), yielding the explicit expression $\|c^\top M^+\|$ for the optimum. So long as we can construct a left inverse $M^+$, this gives an upper bound on $\Corr_{\le D}$. While many choices for $M^+$ are possible, we will need to find one that (i) is simple enough to work with explicitly, and (ii) results in a good bound on $\Corr_{\le D}$ at the end of the day.

\subsection{Computing $M$}

The first step in writing down an explicit expression for $M$ will be to express $T^S$ in the basis $\{A^U\}$.

\begin{definition}[List notation for $S$]
\label{def:list-S}
We will identify $S \in \NN^\Omega$ with a multi-set, namely the multi-set containing $S_I$ copies of each $I \in \Omega$. With some abuse of notation, write $S$ as an ordered list $S = (I_1,\ldots,I_{|S|})$ containing the elements of the associated multi-set sorted according to some arbitrary but fixed ordering on $\Omega$. For a \emph{labeling} $\ell = (\ell_1,\ldots,\ell_{|S|}) \in [r]^{|S|}$, define $S(\ell) \subseteq [n] \times [r]$ to be the subset containing all $(i,j)$ pairs with the following property: the element $i$ occurs in an odd number of the sets $\{I_d \,:\, \ell_d = j\}$. (Informally, $S(\ell)$ is produced by placing each $I_d$ into column $\ell_d$ of an $n \times r$ grid and then XOR'ing the contents of each column.)
\end{definition}

With this notation we can write
\begin{equation}\label{eq:T-exp}
T^S
= \prod_{I \in \Omega} \left(\sum_{j=1}^r \lambda_j a_j^I \right)^{S_I}
= \prod_{d=1}^{|S|} \left(\sum_{\ell_d=1}^r \lambda_{\ell_d} a_{\ell_d}^{I_d} \right)
= \sum_{\ell \in [r]^{|S|}} \lambda^\ell A^{S(\ell)}
\end{equation}
where $\lambda^\ell := \prod_d \lambda_{\ell_d}$.

Next we consider an arbitrary vector of coefficients $\hat f$ and write down an expression for the corresponding $\hat g$ in~\eqref{eq:f-g}. We have
\begin{align*}
    f(T) &= \sum_{|S| \le D} \hat f_S T^S \\
    &= \sum_{|S| \le D} \hat f_S \sum_{\ell \in [r]^{|S|}} \lambda^\ell A^{S(\ell)} \\
    &= \sum_{|U| \le kD} A^U \underbrace{\sum_{|S| \le D} \hat f_S \sum_{\ell \in [r]^{|S|}} \lambda^\ell\, \One_{S(\ell) = U}}_{\hat g_U}.
\end{align*}
In other words, $\hat g = M \hat f$ where
\begin{equation}\label{eq:M-def}
M_{US} := \sum_{\ell \in [r]^{|S|}} \lambda^\ell\, \One_{S(\ell) = U}.
\end{equation}

\subsection{Constructing the Left Inverse}
\label{sec:construct-left}

We now construct a left inverse $M^+$ for $M$, which, recall, is needed to apply the bound~\eqref{eq:corr-bound}. The intuition behind this construction is described in Section~\ref{sec:technique-lower}.

\begin{definition}[Generic $U$]
We call $U \subseteq [n] \times [r]$ \emph{generic} provided that every column $U_j \subseteq [n]$ satisfies $|U_j| \le k$, and at most $D$ columns satisfy $|U_j| > 0$.
\end{definition}

\noindent These are the ``generic'' terms appearing in the expansion~\eqref{eq:T-exp}: $U$ is generic if and only if there exist $S \in \NN^\Omega$ with $|S| \le D$ and $\ell \in [r]^{|S|}$ with \emph{distinct} entries $\ell_1,\ldots,\ell_{|S|}$ such that $S(\ell) = U$. (Note that~\eqref{eq:r-assum-hard} implies $D \le r$, so it is possible for $\ell$ to have distinct entries.) Furthermore, if $U$ is generic then there is a \emph{unique} corresponding $S$, namely $\cols(U)$ defined as follows.

\begin{definition}
For a generic $U$, define $\cols(U) \in \NN^\Omega$ by letting $\cols(U)_I$ be the number of columns $j$ for which $U_j = I$.
\end{definition}

\noindent When viewing $S$ as a multi-set as per Definition~\ref{def:list-S}, $\cols(U)$ is simply the multi-set of columns $U_j$. Recalling the definition $|S| := \sum_{I \in \Omega} S_I$, note that $|\cols(U)|$ denotes the number of non-empty columns of $U$.

Our left inverse $M^+$ will satisfy $M^+_{SU} = 0$ whenever $U$ is not generic; in other words, our procedure for recovering $\hat f$ from $\hat g$ only uses the values $\hat g_U$ for which $U$ is generic. Write $M$ in block form
\[ M = \left[\begin{array}{c} G \\ N \end{array}\right] \]
where $G$ (``generic'') is indexed by generic $U$'s and $N$ (``not'') is indexed by the rest. It suffices to construct a left inverse $G^+$ for $G$ and then set
\begin{equation}\label{eq:M-inv-gen}
M^+ = \left[\begin{array}{cc} G^+ & 0 \end{array}\right].
\end{equation}
Note that $G = G(D)$ has the recursive structure
\[ G(D) = \left[\begin{array}{cc} G(D-1) & R(D) \\ 0 & Q(D) \end{array}\right] \]
where the first block of columns is indexed by $|S| \le D-1$ and the second is indexed by $|S| = D$, and the first block of rows is indexed by $|\cols(U)| \le D-1$ and the second by $|\cols(U)| = D$. Crucially, the lower-left block is 0: recall~\eqref{eq:M-def} and note that $S(\ell) = U$ is only possible when $|S| \ge |\cols(U)|$.
Given any left inverses $G(D-1)^+, Q(D)^+$ for $G(D-1),Q(D)$ respectively, one can verify that the following matrix is a valid left inverse for $G(D)$:
\begin{equation}\label{eq:block-inv}
G(D)^+ := \left[\begin{array}{cc} G(D-1)^+ & -G(D-1)^+ R(D) Q(D)^+ \\ 0 & Q(D)^+ \end{array}\right].
\end{equation}
The left inverse $G(D-1)^+$ can be constructed by applying~\eqref{eq:block-inv} recursively. The matrix $Q(D)$ has only one nonzero entry per row and so we will be able to construct $Q(D)^+$ by hand.

\begin{lemma}
The following is a valid left inverse for $Q = Q(D)$: for $U$ generic and $|S| = |\cols(U)| = D$, define $Q^+ = (Q^+_{SU})$ by
\begin{equation}\label{eq:Q-inv}
Q^+_{SU} := \frac{\One_{\cols(U) = S}}{\lambda^U r^{\underline{|S|}}}
\end{equation}
where
\[ \lambda^U := \prod_{j \,:\, |U_j| > 0} \lambda_j \]
and
\begin{equation}\label{eq:ff}
r^{\underline{d}} := \underbrace{r(r-1)(r-2) \cdots (r-d+1)}_{d \text{ factors}}.
\end{equation}
\end{lemma}

\begin{proof}
We need to verify $Q^+ Q = I$. For $|S| = |S'| = D$,
\[ (Q^+Q)_{SS'} = \sum_{U \,:\, |\cols(U)| = D} Q^+_{SU} Q_{US'}
= \sum_{U \,:\, |\cols(U)| = D} \frac{\One_{\cols(U) = S}}{\lambda^U r^{\underline{|S|}}} \sum_{\ell \in [r]^{|S'|}} \lambda^{\ell}\,\One_{S'(\ell)=U} = \One_{S = S'}, \]
where the last step is justified as follows: since $|S'| = |\cols(U)| = D$, the indicator $\One_{S'(\ell)=U}$ can only be satisfied when $\ell$ has distinct entries. There are $r^{\underline{|S'|}}$ such $\ell$'s, and for each there is exactly one term in the first sum satisfying $\One_{S'(\ell)=U}$, namely $U = S'(\ell)$, and this implies $\lambda^U = \lambda^\ell$ and $S' = \cols(U)$. This means the indicator $\One_{\cols(U) = S}$ becomes $\One_{S = S'}$, and the other factors cancel. This completes the proof.
\end{proof}

\noindent This completes the description of the left inverse $M^+$.

\subsection{Recurrence for $w$}

Ultimately we are interested in an expression not for $M^+$ but for the vector $w^\top := c^\top M^+$ appearing in~\eqref{eq:corr-bound}. We will use the calculations from the previous section to write down a self-contained recursive formula for the entries of $w$. Note that from~\eqref{eq:M-inv-gen}, $w_U = 0$ whenever $U$ is not generic, so we will focus on computing $w_\gen = (w_U \,:\, U \text{ generic})$. Using~\eqref{eq:block-inv},
\[ w_\gen^\top = c^\top G^+ = \left[c^\top \left[\begin{array}{c} G(D-1)^+ \\ 0 \end{array}\right] \quad c^\top \left[\begin{array}{c} -G(D-1)^+ R(D) Q(D)^+ \\ Q(D)^+ \end{array}\right]\right] =: \left[x^\top \quad y^\top\right], \]
where $x$ is indexed by $|\cols(U)| \le D-1$ and $y$ is indexed by $|\cols(U)| = D$. The expression for $x$ reveals that $w_U$ does not depend on $D$ (so long as $D \ge |\cols(U)|$). By comparing the expressions for $x$ and $y$ we can write
\[ y^\top = c^\top \left[\begin{array}{c} 0 \\ Q(D)^+ \end{array}\right] - x^\top R(D) Q(D)^+. \]
For any generic $U$, the case $D = |\cols(U)|$ of the above gives
\[ w_U = \underbrace{\sum_{S \,:\, |S| = |\cols(U)|} c_S Q(D)^+_{SU}}_{\text{(I)}} - \underbrace{\sum_{\substack{\text{generic } U' \\ |\cols(U')| < |\cols(U)|}} \; \sum_{S \,:\, |S| = |\cols(U)|} w_{U'} R(D)_{U'S} Q(D)^+_{SU}}_{\text{(II)}}. \]
We treat the two terms $\text{(I)}, \text{(II)}$ separately. Using the definition~\eqref{eq:Q-inv} for $Q^+$,
\begin{align*}
\text{(I)} &= \sum_{S \,:\, |S| = |\cols(U)|} c_S\, \frac{\One_{\cols(U) = S}}{\lambda^U r^{\underline{|S|}}} \\
&= \frac{c_S}{\lambda^U r^{\underline{|S|}}} \qquad \text{where } S = \cols(U).
\end{align*}
Now for the second term $\text{(II)}$, suppressing the constraints on $U'$ for ease of notation,
\begin{align*}
\text{(II)} &= \sum_{U'} \; \sum_{S \,:\, |S| = |\cols(U)|} w_{U'} R(D)_{U'S} Q(D)^+_{SU} \\
&= \sum_{U'} \; \sum_{S \,:\, |S| = |\cols(U)|} w_{U'} \left(\sum_{\ell \in [r]^{|S|}} \lambda^\ell\, \One_{S(\ell) = U'}\right) \frac{\One_{\cols(U) = S}}{\lambda^U r^{\underline{|S|}}} \\
&= \sum_{U'} w_{U'} \left(\sum_{\ell \in [r]^{|S|}} \lambda^\ell\, \One_{S(\ell) = U'}\right) \frac{1}{\lambda^U r^{\underline{|S|}}} \qquad \text{where } S = \cols(U).
\end{align*}
Putting it together, we have now shown
\begin{equation}\label{eq:corr-w}
\Corr_{\le D}^2 \le \sum_{\substack{\text{generic } U \\ |\cols(U)| \le D}} w_U^2
\end{equation}
where, for generic $U$, the $w_U$ are defined by the recurrence
\begin{equation}\label{eq:w-rec}
w_U = \frac{1}{\lambda^U r^{\underline{|S|}}} \left(c_S - \sum_{\substack{\text{generic } U' \\ |\cols(U')| < |\cols(U)|}} w_{U'} \sum_{\ell \in [r]^{|S|}} \lambda^\ell\,\One_{S(\ell)=U'}\right) \qquad \text{where } S = \cols(U).
\end{equation}

\subsection{Recurrence for $v$}
\label{sec:rec-v}

It will be convenient to rewrite the recurrence~\eqref{eq:w-rec} in terms of a different quantity indexed by $S \in \NN^\Omega$ instead of $U \subseteq [n] \times [r]$, namely
\begin{equation}\label{eq:v-def}
v_S := \lambda^U r^{\underline{|S|}} \, w_U \qquad \text{for $U$ satisfying } \cols(U) = S.
\end{equation}
It can be seen from~\eqref{eq:w-rec} that $v_S$ is well-defined in the sense that it does not depend on the choice of $U$ in~\eqref{eq:v-def}. In particular,
\[ v_S = c_S - \sum_{\substack{\text{generic } U' \\ |\cols(U')| < |S|}} w_{U'} \sum_{\ell \in [r]^{|S|}} \lambda^\ell\,\One_{S(\ell)=U'}. \]
Using~\eqref{eq:v-def} we can turn this into a self-contained recurrence for $v$ (not involving $w$):
\begin{align}
v_S &= c_S - \sum_{\substack{\text{generic } U' \\ |\cols(U')| < |S|}} \frac{v_{S'}}{\lambda^{U'} r^{\underline{|S'|}}} \sum_{\ell \in [r]^{|S|}} \lambda^\ell\,\One_{S(\ell)=U'} \qquad \text{where } S' = \cols(U') \nonumber\\
&= c_S - \sum_{\substack{S' \in \NN^\Omega \\ |S'| < |S|}} v_{S'} \sum_{\ell \in [r]^{|S|}} \frac{\lambda^\ell}{\lambda^{S(\ell)} r^{\underline{|S'|}}} \, \One_{\cols(S(\ell)) = S'}
\label{eq:v-rec}
\end{align}
where the predicate $\cols(S(\ell)) = S'$ in particular requires $S(\ell)$ to be generic. For any $S \in \NN^\Omega$ there are at most $r^{\underline{|S|}}$ corresponding generic $U$'s for which $S = \cols(U)$. For any such $U$, we have from~\eqref{eq:v-def} that $|w_U| \le |v_S|/(\lambda_{\min}^{|S|}\, r^{\underline{|S|}})$. This means~\eqref{eq:corr-w} gives
\begin{equation}\label{eq:corr-v}
\Corr_{\le D}^2 \le \sum_{\substack{S \in \NN^\Omega \\ |S| \le D}} r^{\underline{|S|}} \left(\frac{v_S}{\lambda_{\min}^{|S|}\, r^{\underline{|S|}}}\right)^2 = \sum_{\substack{S \in \NN^\Omega \\ |S| \le D}} \frac{v_S^2}{\lambda_{\min}^{2|S|}\, r^{\underline{|S|}}}
\end{equation}
where $v_S$ is defined by the recurrence~\eqref{eq:v-rec}.

\subsection{Grouping by Patterns}

The core challenge remaining is to analyze the recurrence~\eqref{eq:v-rec} and establish an upper bound on $|v_S|$. Naively bounding the terms in~\eqref{eq:v-rec} will not suffice, as there are subtle cancellations that occur. To understand the nature of these cancellations, we will rewrite~\eqref{eq:v-rec} in a different form. First we will group the terms in~\eqref{eq:v-rec} by their type, defined as followed. We use $\oplus$ to denote the XOR (symmetric difference) operation on sets.

\begin{definition}\label{def:pi}
Fix $S \in \NN^\Omega$, viewed as a list $S = (I_1,\ldots,I_{|S|})$ as per Definition~\ref{def:list-S}. Define a \emph{pattern} $\pi = (\pi_1,\ldots,\pi_{|S|})$ to be an element of $([r] \cup \{\star\})^{|S|}$ satisfying the following rules:
\begin{itemize}
    \item[(i)] ``No singletons'': if $\pi_d = j \in [r]$ then there must exist $d' \ne d$ such that $\pi_{d'} = j$.
    \item[(ii)] ``Not all stars'': there must exist $d$ such that $\pi_d \ne \star$.
    \item[(iii)] ``Every column valid'': for every $j \in [r]$, we have $|\oplus_{d \,:\, \pi_d = j} I_d| \le k$.
\end{itemize}
Let $\Pi(S)$ denote the set of all such patterns. We let $S(\pi)_j := \oplus_{d \,:\, \pi_d = j} I_d$ and define $S(\pi) \subseteq [n] \times [r]$ to have $j$th column $S(\pi)_j$ for all $j \in [r]$.
\end{definition}

Note that if $\pi$ has no stars, it is simply a labeling $\ell \in [r]^{|S|}$, in which case the definitions $S(\pi)$ and $S(\ell)$ coincide. Intuitively, a pattern $\pi$ describes a class of possible labelings $\ell$, where the stars are ``wildcards'' that may represent any element of $[r]$ (subject to some restrictions). More formally, we now define which $\ell$'s belong to a pattern $\pi$.

\begin{definition}\label{def:ell-in-pi}
Fix $S \in \NN^\Omega$, viewed as a list $S = (I_1,\ldots,I_{|S|})$. For $\ell \in [r]^{|S|}$ and $\pi \in \Pi(S)$, write $\pi \vdash \ell$ if the following conditions hold:
\begin{itemize}
    \item[(i)] For each $d$, if $\pi_d = j \in [r]$ then $\ell_d = j$.
    \item[(ii)] The ``starred'' columns $\ell_\star := \{\ell_d \,:\, \pi_d = \star\}$ are distinct.
    \item[(iii)] For every $j \in \ell_\star$ we have $S(\pi)_j = \emptyset$.
\end{itemize}
\end{definition}

Let $\mathcal{L}(S)$ denote the set of labelings $\ell \in [r]^{|S|}$ for which there exists $S' \in \NN^\Omega$ such that $|S'| < |S|$ and $\cols(S(\ell)) = S'$; these are the $\ell$'s that contribute to the sum in~\eqref{eq:v-rec}. For any $\ell \in \mathcal{L}(S)$, there is at least one (and possibly more than one) pattern $\pi \in \Pi(S)$ for which $\pi \vdash \ell$. The following inclusion-exclusion formula will allow us to sum over $\pi \in \Pi(S)$ in a way that counts every $\ell \in \mathcal{L}(S)$ exactly once.

\begin{lemma}\label{lem:inc-exc}
Fix $S \in \NN^\Omega$, viewed as a list $S = (I_1,\ldots,I_{|S|})$. Fix a function $\phi: [r]^{|S|} \to \RR$. For $\pi \in \Pi(S)$ and $j \in [r]$, define
\[ m_\pi(j) = |\{d \,:\, \pi_d = j \text{ and } I_d = S(\pi)_j\}| \]
and
\[ m_\pi = \prod_{j \in [r]} (1 - m_\pi(j)). \]
Then we have
\begin{equation}\label{eq:inc-exc}
\sum_{\ell \in \mathcal{L}(S)} \phi(\ell) = \sum_{\pi \in \Pi(S)} m_\pi \sum_{\ell \in [r]^{|S|} \,:\, \pi \vdash \ell} \phi(\ell).
\end{equation}
\end{lemma}

\begin{proof}
First note that $\pi \vdash \ell$ implies $\ell \in \mathcal{L}(S)$ and so there are no ``extra'' terms $\phi(\ell)$ on the right-hand side that are not present on the left-hand side. For any fixed $\ell \in \mathcal{L}(S)$, the term $\phi(\ell)$ appears exactly once on the left-hand side of~\eqref{eq:inc-exc}. Our goal is to show that it also appears exactly once on the right-hand side, that is, it suffices to prove
\[ \sum_{\pi \in \Pi(S) \,:\, \pi \vdash \ell} m_\pi = 1. \]
For a fixed $\ell$, we need to enumerate the possible patterns $\pi$ for which $\pi \vdash \ell$. The rules for these patterns are as follows:
\begin{itemize}
    \item (Case 1) For any $j$, if there is exactly one $d$ for which $\ell_d = j$ then $\pi_d = \star$. In this case, $m_\pi(j) = 0$.
    \item (Case 2) For any $j$, if $S(\ell)_j = \emptyset$ then there are no stars among $\{\pi_d \,:\, \ell_d = j\}$. In this case, $m_\pi(j) = 0$.
    \item (Case 3) For any $j$ not in Case 1, if $S(\ell)_j \ne \emptyset$ then among $\{\pi_d \,:\, \ell_d = j\}$ there are either no stars or exactly one star of the form $\pi_d = \star$ where $\ell_d = j$ and $I_d = S(\ell)_j$. If there are no stars then $m_\pi(j) = m_\ell(j) := |\{d \,:\, \ell_d = j \text{ and } I_d = S(\ell)_j\}|$. There are $m_\ell(j)$ ways to have one star, and each yields $m_\pi(j) = 0$.
\end{itemize}
Now we have
\[ \sum_{\pi \,:\, \pi \vdash \ell} m_\pi = \sum_{\pi \,:\, \pi \vdash \ell} \prod_{j \in [r]} (1 - m_\pi(j))
= \prod_{j \text{ in Case 3}} [(1-m_\ell(j)) + m_\ell(j) (1-0)] = 1 \]
as desired.
\end{proof}

Using Lemma~\ref{lem:inc-exc}, we can rewrite~\eqref{eq:v-rec} as
\begin{align*}
v_S &= c_S - \sum_{S' \,:\, |S'| < |S|} v_{S'} \sum_{\ell \in [r]^{|S|}} \frac{\lambda^\ell}{\lambda^{S(\ell)} r^{\underline{|S'|}}} \, \One_{\cols(S(\ell)) = S'} \\
&= c_S - \sum_{S' \,:\, |S'| < |S|} v_{S'} \sum_{\pi \in \Pi(S)} m_\pi \sum_{\ell \,:\, \pi \vdash \ell} \frac{\lambda^\ell}{\lambda^{S(\ell)} r^{\underline{|S'|}}} \, \One_{\cols(S(\ell)) = S'}.
\end{align*}
The number of labelings $\ell$ such that $\pi \vdash \ell$ is $(r_\pi)^{\underline{s_\pi}}$ where $r_\pi$ is the number of columns $j \in [r]$ for which $S(\pi)_j = \emptyset$, and $s_\pi$ is the number of stars in $\pi$ (i.e., the number of indices $d$ such that $\pi_d = \star$). Note that the ratio $\frac{\lambda^\ell}{\lambda^{S(\ell)}}$ depends only on $\pi$ (not on $\ell$), and so we define $\lambda^{(\pi)} := \frac{\lambda^\ell}{\lambda^{S(\ell)}}$. Also, the predicate $\cols(S(\ell)) = S'$ depends only on $\pi$ (not on $\ell$), and we write this predicate as $S \overset{\pi}{\longrightarrow} S'$. With this notation, the recurrence for $v_S$ becomes
\begin{equation}\label{eq:v-rec-pi}
v_S = c_S - \sum_{S' \,:\, |S'| < |S|} v_{S'} \sum_{\pi \in \Pi(S)} m_\pi\, (r_\pi)^{\,\underline{s_\pi}}\, \frac{\lambda^{(\pi)}}{r^{\underline{|S'|}}} \, \One_{S \overset{\pi}{\longrightarrow} S'}.
\end{equation}

\subsection{Unravelling the Recurrence}

Next we will rewrite~\eqref{eq:v-rec-pi} in closed form (without recursion) by expanding the recursion tree as a sum over ``paths.''

We first unpack the definition~\eqref{eq:c-def} for $c_S$. Using~\eqref{eq:T-exp} and recalling that $A$ has i.i.d.\ Rademacher entries,
\begin{equation}\label{eq:c-exp}
c_S = \EE[T^S \cdot a_{11}] = \sum_{\ell \in [r]^{|S|}} \lambda^\ell \,\EE [A^{S(\ell)} \cdot a_{11}] = \sum_{\ell \in [r]^{|S|}} \lambda^\ell \,\One_{S(\ell) = \{(1,1)\}}.
\end{equation}
By expanding the recursion tree of~\eqref{eq:v-rec-pi} we can write $v_S$ as a sum over paths which we denote by
\begin{equation}\label{eq:path}
S = S^0 \overset{\pi^0}{\longrightarrow} S^1 \overset{\pi^1}{\longrightarrow} \cdots \overset{\pi^{p-1}}{\longrightarrow} S^p \overset{\pi^p}{\longrightarrow} \,\perp.
\end{equation}
Formally, a path consists of $S^0$ and $\pi^1,\ldots,\pi^p$ (which then determine $S^1,\ldots,S^p$) with $S^t \in \NN^\Omega$ for all $t$, $\pi^t \in \Pi(S^t)$ for $t \le p-1$, and $\pi^p \in [r]^{|S^p|}$ (``no stars at the final step'') such that the following properties hold:
\begin{itemize}
    \item For all $0 \le t \le p-1$, we require $|S^{t+1}| < |S^t|$ and the predicate $S^t \overset{\pi^t}{\longrightarrow} S^{t+1}$ holds.
    \item For the final step $S^p \overset{\pi^p}{\longrightarrow} \,\perp$ we require $S^p(\pi^p) = \{(1,1)\}$.
\end{itemize}
With this notation,~\eqref{eq:v-rec-pi} can be written as
\begin{equation}\label{eq:v-rec-expanded}
v_S = \sum_{p \ge 0} \; \sum_{S = S^0 \overset{\pi^0}{\longrightarrow} S^1 \overset{\pi^1}{\longrightarrow} \, \cdots \, \overset{\pi^{p-1}}{\longrightarrow} S^p \overset{\pi^p}{\longrightarrow} \,\perp} (-1)^p \left(\prod_{t=0}^{p-1} m_{\pi^t}\, (r_{\pi^t})^{\,\underline{s_{\pi^t}}}\, \frac{\lambda^{(\pi^t)}}{r^{\underline{|S^{t+1}|}}}\right) \lambda^{\pi^p},
\end{equation}
where the special rule for the final step comes from~\eqref{eq:c-exp}.

\subsection{Excluding Bad Paths}

The next step is the crux of the argument: we will show that only certain types of paths contribute to~\eqref{eq:v-rec-expanded}, due to cancellations among the remaining paths.

\begin{definition}[Event]
For a path of the form~\eqref{eq:path}, we say an \emph{event} occurs at timestep $t \in \{0,1,\ldots,p\}$ on column $j \in [r]$ if there exists $d$ for which $\pi^t_d = j$.
\end{definition}

\noindent Note that Definition~\ref{def:pi} requires every timestep $t$ to have an event on at least one column $j$.

\begin{definition}[Deletion event]
An event at timestep $t$ on column $j$ is called a \emph{deletion event} if $S^t(\pi^t)_j = \emptyset$.
\end{definition}

\begin{definition}[Good/bad paths]
A path of the form~\eqref{eq:path} is called \emph{bad} if there exists a column $j \in [r]$ such that the last event (i.e., with the largest $t$) on that column is a deletion event. If a path is not bad, it is called \emph{good}.
\end{definition}

\begin{lemma}
The total contribution from bad paths to~\eqref{eq:v-rec-expanded} is 0. That is,
\begin{equation}\label{eq:v-rec-good}
v_S = \sum_{p \ge 0} \; \sum_{\substack{S = S^0 \overset{\pi^0}{\longrightarrow} S^1 \overset{\pi^1}{\longrightarrow} \, \cdots \, \overset{\pi^{p-1}}{\longrightarrow} S^p \overset{\pi^p}{\longrightarrow} \,\perp \\ \text{good}}} (-1)^p \left(\prod_{t=0}^{p-1} m_{\pi^t}\, (r_{\pi^t})^{\,\underline{s_{\pi^t}}}\, \frac{\lambda^{(\pi^t)}}{r^{\underline{|S^{t+1}|}}}\right) \lambda^{\pi^p}.
\end{equation}
\end{lemma}

\begin{proof}
We will show that the bad paths can be paired up so that within each pair, the two paths contribute the same term to~\eqref{eq:v-rec-expanded} but with opposite signs. The pairing is described by the following procedure, an involution that maps each bad path to its partner (and vice versa).
\begin{enumerate}
    \item[(1)] Given a bad path as input, let $j^*$ denote the largest column index for which the last event is a deletion event. Let $t^*$ denote the timestep on which this deletion event occurs.
    \item[(2a)] If there exists another event at timestep $t^*$ (on some column $j \ne j^*$), ``promote'' the $(t^*,j^*)$ deletion event to its own timestep. That is, replace
    \[ \cdots \; S^{t^*} \overset{\pi^{t^*}}{\longrightarrow} \cdots \qquad\text{by}\qquad \cdots \; S^{t^*} \overset{\tau}{\longrightarrow} S' \overset{\sigma}{\longrightarrow} \cdots \]
    where $\tau, \sigma$ are defined as follows. For all $d$ such that $\pi^{t^*}_d = j^*$, set $\tau_d = j^*$; for all other $d$, set $\tau_d = \star$. Now $S'$ (viewed as a list) is produced from $S^{t^*}$ by removing the elements $I_d$ for which $\pi^{t^*}_d = j^*$; similarly, $\sigma$ is produced from $\pi^{t^*}$ by removing the elements $\pi^{t^*}_d$ that are equal to $j^*$.
    \item[(2b)] If instead there is no other event at timestep $t^*$ (which cannot happen if $t^* = p$ due to the non-deletion event on column 1), ``merge'' timestep $t^*$ with the subsequent timestep. That is, replace
    \[ \cdots \; S^{t^*} \overset{\pi^{t^*}}{\longrightarrow} S^{t^*+1} \overset{\pi^{t^*+1}}{\longrightarrow} \cdots \qquad\text{by}\qquad \cdots \; S^{t^*} \overset{\tau}{\longrightarrow} \cdots \]
    where $\tau$ is defined as follows. For all $d$ such that $\pi^{t^*}_d = j^*$, set $\tau_d = j^*$. The number of remaining entries of $\tau$ is exactly $|S^{t^*+1}|$, and we designate these entries as ``unassigned''; for each $1 \le i \le |S^{t^*+1}|$, set the $i$th \emph{unassigned} entry of $\tau$ to be $\pi^{t^*+1}_i$. Note that since the $(t^*,j^*)$ event was the last event in its column, the merge operation will not cause it to ``collide'' with another event.
\end{enumerate}
A few claims remain to be checked before the proof is complete. First note that the above procedure maps any bad path to a different bad path (its ``partner''), and applying the procedure again on the partner recovers the original path. For instance, if applying the procedure to the original path resulted in a ``promote'' operation on column $j^*$, applying the procedure on the partner will undo this using a ``merge'' operation on the same column $j^*$.

We furthermore claim that for any bad path and its partner, both paths have the same value for the factor $\left(\prod_{t=0}^{p-1} \cdots \right) \lambda^{\pi^p}$ in~\eqref{eq:v-rec-expanded}. However, the lengths of the two paths differ by one, causing the two corresponding terms in~\eqref{eq:v-rec-expanded} to cancel due to the $(-1)^p$ factor. To prove the claim, compare the cases $S^t \overset{\pi^t}{\longrightarrow} S^{t+1}$ and $S^t \overset{\tau}{\longrightarrow} S' \overset{\sigma}{\longrightarrow} S^{t+1}$ from (2a) above, where for now we assume $S^{t+1} \ne \perp$. For the $m$ factors, note that $m_\pi(j) = 0$ whenever $\pi$ has either a deletion event on column $j$ or no event on column $j$, and so $m_{\pi^t} = m_\tau m_\sigma$. For the $r^{\underline{s}}$ factors, note that $r_{\pi^t} = r_{\sigma}$ and $s_{\pi^t} = s_\sigma$; also, $r_\tau = r$ and $s_\tau = |S'|$, so $(r_\tau)^{\underline{s_\tau}}$ cancels with the existing factor of $1/r^{\underline{|S'|}}$ in~\eqref{eq:v-rec-expanded}. Finally, for the $\lambda$ factors we have $\lambda^{(\pi^t)} = \lambda^{(\tau)} \lambda^{(\sigma)}$. The other case $S^t \overset{\pi^t}{\longrightarrow} \, \perp$ versus $S^t \overset{\tau}{\longrightarrow} S' \overset{\sigma}{\longrightarrow} \, \perp$ is treated similarly, where now we have $m_\tau = 1$, $(r_\tau)^{\underline{s_\tau}} = r^{\underline{|S'|}}$, and $\lambda^{\pi^t} = \lambda^{(\tau)} \lambda^{\sigma}$. This completes the proof.
\end{proof}

\subsection{Bounding $|v_S|$}

Now that we have identified the crucial cancellations between pairs of bad paths, the rest of the proof will follow by bounding the terms in~\eqref{eq:v-rec-good} in a straightforward way. We start by collecting some bounds on the individual pieces of~\eqref{eq:v-rec-good}.

\begin{lemma}\label{lem:v-help-1}
For any step $S \overset{\pi}{\longrightarrow} S'$, $|m_\pi| \le 2^{|S| - |S'|}$.
\end{lemma}

\begin{proof}
Note that
\[ |m_\pi(j)| \le |\{d \,:\, \pi_d = j\}| =: |\pi^{-1}(j)|. \]
The number of stars plus the number of distinct columns in $\pi$ must be at least $|S'|$. This leaves at most $|S|-|S'|$ entries of $\pi$ that repeat a previous column, i.e.,
\begin{equation}\label{eq:sumSS}
\sum_{j \in [r] \,:\, |\pi^{-1}(j)| \ge 2} (|\pi^{-1}(j)|-1) \le |S| - |S'|.
\end{equation}
This means
\begin{align*}
|m_\pi| = \prod_{j \in [r]} |m_\pi(j) - 1|
&\le \prod_{j \in [r] \,:\, |m_\pi(j)| \ge 2} (m_\pi(j) - 1)
\le \prod_{j \in [r] \,:\, |\pi^{-1}(j)| \ge 2} (|\pi^{-1}(j)| - 1) \\
&\le \prod_{j \in [r] \,:\, |\pi^{-1}(j)| \ge 2} 2^{(|\pi^{-1}(j)| - 1)} = 2^{\sum_{j \in [r] \,:\, |\pi^{-1}(j)| \ge 2} (|\pi^{-1}(j)| - 1)} \\
&\le 2^{|S| - |S'|}
\end{align*}
where the final step uses~\eqref{eq:sumSS}.
\end{proof}

\begin{lemma}\label{lem:v-help-2}
For any step $S \overset{\pi}{\longrightarrow} S'$, $|\lambda^{(\pi)}| \le 1$.
\end{lemma}

\begin{proof}
Recall $\lambda^{(\pi)} := \frac{\lambda^\ell}{\lambda^{S(\ell)}}$ for any $\ell$ such that $\pi \vdash \ell$. Recall that $\lambda^{S(\ell)}$ is the product of $\lambda_j$ over the non-empty columns $j$ of $S(\ell)$. For any such non-empty column $j$, there must exist $d$ for which $\ell_d = j$. Thus, every factor of $\lambda_j$ in the denominator of $\frac{\lambda^\ell}{\lambda^{S(\ell)}}$ is cancelled by the numerator, and the result now follows because $|\lambda_j| \le 1$.
\end{proof}

\noindent We now state the main conclusion of this section.

\begin{lemma}\label{lem:v-bound}
For any $S \in \NN^\Omega$ with $|S| \ge 1$, we have $|v_S| \le (3|S|^2)^{|S|}$.
\end{lemma}

\noindent Note that for $|S| = 0$, it can be verified directly that $v_\emptyset = 0$.

\begin{proof}
Proceed by induction on $|S|$. We will bound the sum of absolute values of terms in~\eqref{eq:v-rec-good}; it will no longer be important to exploit cancellations between positive and negative terms. First consider paths of the form $S \overset{\pi^0}{\longrightarrow} \,\perp$. There is at most one such path that is good, namely $\pi^0 = (1,1,\ldots,1)$; the value of this term is $|\lambda^{\pi^0}| \le 1$.

All remaining paths take the form $S \overset{\pi^0}{\longrightarrow} S^1 \cdots$ where $|S^1| = i$ for some $1 \le i \le |S|-1$. In order to produce $S^1$, $\pi^0$ must have exactly $i$ entries $d$ that are either stars (i.e., $\pi^0_d = \star$) or first in a ``combination'' event (i.e., for some $j \in [r]$ with $S(\pi^0)_j \ne \emptyset$, $d$ is the lowest index such that $\pi^0_d = j$). There are $\binom{|S|}{i}$ choices for which $i$ elements of $\pi^0$ play these roles; by default they will all be stars, and will be converted to ``combinations'' if another entry of $\pi^0$ decides to join them on the same column.

Now there are $|S|-i$ entries of $\pi^0$ remaining, which have a few options. One option is to participate in a combination event by joining one of the $i$ previously designated entries on the same column. The other option is to participate in a deletion event. Since we are only counting good paths, this can only happen on a column on which a later event will occur. Regardless of the remainder of the path $S^1 \overset{\pi^1}{\longrightarrow} \cdots \,\perp$, there are at most $|S^1| = i$ such columns available. Since each of $|S|-i$ entries of $\pi^0$ has at most $2i$ choices, this gives a total of at most $(2i)^{|S|-i}$ choices.

We also need to decide which columns the combination events occur on. If $\pi^0$ has $s_{\pi^0}$ stars and $i - s_{\pi^0}$ combination events, there are $r^{\underline{i - s_{\pi^0}}}$ choices for the columns. Note that this exactly cancels the factor $(r_{\pi^0})^{\underline{s_{\pi^0}}}/r^{\underline{|S^1|}}$ in~\eqref{eq:v-rec-good} because $r_{\pi^0} = r - (i - s_{\pi^0})$ and $|S^1| = i$.

Recall that we aim to show $|v_S| \le b(|S|)$ where $b(d) := (3d^2)^d$. Plugging the above calculations (along with Lemmas~\ref{lem:v-help-1} and~\ref{lem:v-help-2}) into~\eqref{eq:v-rec-good} and using the induction hypothesis to handle the remainder of the path $S^1 \overset{\pi^1}{\longrightarrow} \cdots \,\perp$, we have
\begin{align*}
|v_S| &\le 1 + \sum_{i=1}^{|S|-1} b(i) \binom{|S|}{i} (2i)^{|S|-i} \, 2^{|S|-i} \\
&= 1 + \sum_{i=1}^{|S|-1} (3i^2)^i \binom{|S|}{i} (4i)^{|S|-i}.
\intertext{At this point we can verify the case $|S| = 1$ directly. Assuming $|S| \ge 2$, we continue:}
&\le 1 + \sum_{i=1}^{|S|-1} \binom{|S|}{i} (3(|S|-1)^2)^i (4(|S|-1))^{|S|-i} \\
&= 1 + \sum_{i=0}^{|S|} \binom{|S|}{i} (3(|S|-1)^2)^i (4(|S|-1))^{|S|-i} - (4(|S|-1))^{|S|} - (3(|S|-1)^2)^{|S|} \\
&\le \sum_{i=0}^{|S|} \binom{|S|}{i} (3(|S|-1)^2)^i (4(|S|-1))^{|S|-i} \\
&= [3(|S|-1)^2 + 4(|S|-1)]^{|S|} \\
&= (3|S|^2 - 6|S| + 3 + 4|S|-4)^{|S|} \\
&\le (3|S|^2)^{|S|},
\end{align*}
completing the proof.
\end{proof}

\subsection{Putting it Together}

We now combine~\eqref{eq:corr-v} with Lemma~\ref{lem:v-bound} to complete the proof of Theorem~\ref{thm:main-lower}. We first note that $v_S \ne 0$ only when the elements of $S$ (viewed as a multi-set) together with $\{1\}$ form an even cover of $[n]$.

\begin{lemma}\label{lem:even-cover}
Let $P(S)$ denote the property that $S(\ell) = \{(1,1)\}$ for $\ell = (1,1,\ldots,1)$. If $P(S)$ fails to hold then $v_S = 0$.
\end{lemma}

\begin{proof}
From~\eqref{eq:c-exp}, note that if $c_S \ne 0$ then $S(\ell) = \{(1,1)\}$ for some $\ell$, which implies $S(\ell) = \{(1,1)\}$ for $\ell = (1,1,\ldots,1)$. Thus $c_S = 0$ whenever $P(S)$ fails.

Also note that if $P(S)$ fails and $\cols(S(\ell)) = S'$ for some $\ell$, then $P(S')$ fails. The result now follows from~\eqref{eq:v-rec} using induction on $|S|$.
\end{proof}

\begin{lemma}\label{lem:count-even-cover}
For any $d \ge 1$, the number of multi-sets $S \in \NN^\Omega$ with $|S| = d$ such that $P(S)$ holds is at most $n^{(kd-1)/2} ((kd+3)/2)^{kd}$.
\end{lemma}

\begin{proof}
Since $P(S)$ holds, $S$ together with $\{1\}$ forms an even cover of $[n]$. Therefore the number of elements of $[n] \setminus \{1\}$ covered by $S$ is at most $(kd-1)/2$. The number of ways to choose this many elements is at most $n^{(kd-1)/2}$. Once these are chosen, $S$ has at most $(kd-1)/2 + 1 = (kd+1)/2$ elements to draw from, so the number of possibilities for $S$ is at most $((kd+1)/2+1)^{kd} = ((kd+3)/2)^{kd}$.
\end{proof}

\begin{proof}[Proof of Theorem~\ref{thm:main-lower}]
Starting from~\eqref{eq:corr-v} and using Lemmas~\ref{lem:v-bound}, \ref{lem:even-cover}, and \ref{lem:count-even-cover},
\begin{align*}
\Corr_{\le D}^2 &\le \sum_{|S| \le D} \frac{v_S^2}{\lambda_{\min}^{2|S|}\, r^{\underline{|S|}}} \\
&\le \sum_{d=1}^D n^{(kd-1)/2} ((kd+3)/2)^{kd} \frac{(3d^2)^{2d}}{\lambda_{\min}^{2d} (r-d+1)^d} \\
&= n^{-1/2} \sum_{d=1}^D \left(\frac{n^{k/2} ((kd+3)/2)^k (3d^2)^2}{\lambda_{\min}^2 (r-d+1)}\right)^d \\
&\le n^{-1/2} \sum_{d=1}^D \left(9D^4 (k(D+1)/2)^k \cdot \frac{n^{k/2} }{\lambda_{\min}^2 (r-D)}\right)^d \\
&\le n^{-1/2} \sum_{d=1}^D \left(9D^4 (kD)^k \cdot \frac{19 n^{k/2} }{18 \lambda_{\min}^2 r}\right)^d
\intertext{since~\eqref{eq:r-assum-hard} implies $D \le r/19$}
&= n^{-1/2} \sum_{d=1}^D \left(\frac{19}{2}\, k^k D^{k+4} \cdot \frac{n^{k/2} }{\lambda_{\min}^2 r}\right)^d \\
&\le n^{-1/2} \sum_{d=1}^D \left(\frac{1}{2}\right)^d
\intertext{where we have used the assumption~\eqref{eq:r-assum-hard}}
&\le n^{-1/2}.
\end{align*}
Using Fact~\ref{fact:corr-mmse}, this completes the proof.
\end{proof}

\section{Proof of Theorem~\ref{thm:main-upper}: Upper Bound}

\subsection{Expander Graphs}

We begin by collecting some standard properties of expander graphs.

\begin{proposition}\label{prop:expander}
Fix an integer $k \ge 3$. For all even $N$ exceeding a constant $N_0 = N_0(k)$, there exists a $k$-regular $N$-vertex (simple) graph with the following properties:
\begin{itemize}
    \item the minimum cut value is $k$ (achieved by a single vertex), and
    \item for any $S \subseteq V(G)$ with $0 < |S| \le N/2$, $|\partial S| \ge c |S|$,
\end{itemize}
where $\partial S$ is the set of edges with exactly one endpoint in $S$, and $c \ge 0.08$ is an absolute constant.
\end{proposition}

\begin{proof}
It follows from classical results that a uniformly random $k$-regular $N$-vertex graph has these properties with high probability, i.e., probability $1-o(1)$ as $N \to \infty$ with $k$ fixed. Letting $G$ be such a graph, it is well-known that $G$ is $k$-connected with high probability~\cite[Section~7.6]{boll-book}, which proves the first statement about the minimum cut.

For the second statement, let $k = \mu_1 \ge \mu_2 \ge \cdots \ge \mu_N$ denote the eigenvalues of the adjacency matrix of $G$. Friedman's second eigenvalue theorem~\cite{friedman} states that for any fixed $\epsilon > 0$, $\mu_2 \le 2\sqrt{k-1} + \epsilon$ with high probability. Cheeger's inequality~\cite{dodz,AM-cheeger} (see Theorem~2.4 of~\cite{HLW-cheeger}) tells us that for any $S \subseteq V(G)$ with $0 < |S| \le N/2$, $\frac{|\partial S|}{|S|} \ge \frac{1}{2} (k - \mu_2)$. Combining these gives
\[ |\partial S| \ge \frac{|S|}{2}(k - \mu_2) \ge \frac{|S|}{2}(k - 2\sqrt{k-1} - \epsilon) \]
which concludes the proof for any choice of $c$ satisfying $0 < c < \frac{1}{2}(k - 2\sqrt{k-1})$. The expression $\frac{1}{2}(k - 2\sqrt{k-1})$ is minimized (over $k \ge 3$) when $k = 3$.
\end{proof}

\subsection{Constructing the Polynomial}\label{sec:construct-poly}

Let $N = D-1$ where $D$ is defined as in Theorem~\ref{thm:main-upper}, choosing $n_0$ large enough so that $N \ge N_0$. From this point onward, let $G$ denote the $k$-regular $N$-vertex graph guaranteed by Proposition~\ref{prop:expander}. We construct a new graph $H$ as follows. Starting with $G$, add two additional vertices called $\circ$ and $u$, and add the edge $(\circ, u)$. Recall that $k'$ is the odd element of $\{k-1,k\}$. Choose $p := (k'-1)/2$ arbitrary edges $(i_1,j_1),\ldots,(i_p,j_p)$ of $G$ with no endpoints in common. Delete these $p$ edges and add the edges $(u,i_1),(u,j_1),\ldots,(u,i_p),(u,j_p)$. This completes the description of $H$. Note that $H$ is $k$-regular aside from the degree-1 vertex $\circ$ and the degree-$k'$ vertex $u$.

\begin{definition}\label{def:edge-labeling}
Define an \emph{edge-labeling of $H$} to be a function $\phi: E(H) \to [n]$ that is injective (no two edges get the same label) with $\phi(\circ,u) = 1$. Let $\Phi$ denote the set of all edge-labelings of $H$.
\end{definition}

For an edge-labeling $\phi$ and a vertex $v \in V(H) \setminus \{\circ\}$, define $T_v(\phi)$ to be the following entry of the input tensor: let $e_1,\ldots,e_m$ be the edges incident to $v$ (where $m \in \{k,k'\}$) and then let $T_v(\phi) := T_{\phi(e_1),\ldots,\phi(e_m)}$. Our polynomial estimator is defined as follows:
\begin{equation}\label{eq:f-def}
f(T) = \frac{1}{|\Phi|} \sum_{\phi \in \Phi} \, \prod_{v \in V(H) \setminus \{\circ\}} T_v(\phi).
\end{equation}

\subsection{Vertex Labelings}

\begin{definition}
Define a \emph{vertex-labeling of $H$} to be a function $\psi: V(H) \setminus \{\circ\} \to [r]$. Let $\Psi$ denote the set of all vertex-labelings of $H$.
\end{definition}

For $\phi \in \Phi$, $\psi \in \Psi$, and $v \in V(H) \setminus \{\circ\}$, define $T_v(\phi,\psi)$ as follows: letting $e_1,\ldots,e_m$ be the edges incident to $v$, and $j := \psi(v)$, let $T_v(\phi,\psi) := \lambda_j \prod_{i=1}^m (a_j)_{\phi(e_i)}$. Recalling~\eqref{eq:T_I}, we can expand~\eqref{eq:f-def} as
\begin{align*}
f(T) &= \frac{1}{|\Phi|} \sum_{\phi \in \Phi} \, \sum_{\psi \in \Psi} \, \prod_{v \in V(H) \setminus \{\circ\}} T_v(\phi,\psi) \\
&= \sum_{\psi \in \Psi} \frac{1}{|\Phi|} \sum_{\phi \in \Phi} \, \prod_{v \in V(H) \setminus \{\circ\}} T_v(\phi,\psi).
\end{align*}
\noindent We will break the above sum into different terms depending on $\psi$. Define the partition $\Psi = \Psi_1 \sqcup \Psi_2 \sqcup \Psi_3$ as follows:
\begin{itemize}
\item $\Psi_1 = \{\psi_1\}$ where $\psi_1$ denotes the all-ones labeling: $\psi_1(v) = 1$ for all $v$,
\item $\Psi_2 = \{\psi_2,\ldots,\psi_r\}$ where $\psi_j$ denotes the all-$j$'s labeling: $\psi_j(v) = j$ for all $v$,
\item $\Psi_3 = \Psi \setminus (\Psi_1 \cup \Psi_2)$.
\end{itemize}
\noindent We can now write $f = f_1 + f_2 + f_3$ where for $i \in \{1,2,3\}$,
\[ f_i := \sum_{\psi \in \Psi_i} \frac{1}{|\Phi|} \sum_{\phi \in \Phi} \, \prod_{v \in V(H) \setminus \{\circ\}} T_v(\phi,\psi). \]

\subsection{Signal Term}

We first handle the terms $f_1$ and $f_2$.

\begin{lemma}\label{lem:f1}
$f_1 = a_{11}$.
\end{lemma}

\begin{proof}
We have
\[ f_1 = \frac{1}{|\Phi|} \sum_{\phi \in \Phi} \, \prod_{v \in V(H) \setminus \{\circ\}} T_v(\phi,\psi_1). \]
Recalling that $a_j$ has $\{\pm 1\}$-valued entries and $\lambda_1 = 1$, note that for any $\phi \in \Phi$,
\[ \prod_{v \in V(H) \setminus \{\circ\}} T_v(\phi,\psi_1) = a_{11}, \]
because each edge $e \in E(H)$ contributes a factor of $(a_1)_{\phi(e)}^2 = 1$ except the edge $(\circ,u)$, which is required by Definition~\ref{def:edge-labeling} to have $\phi(\circ,u) = 1$ and thus contributes a factor of $a_{11}$. The result follows.
\end{proof}

\begin{lemma}\label{lem:f2}
$\EE[f_2^2] \le (r-1) \cdot|\lambda_2|^{2(N+1)}$.
\end{lemma}

\begin{proof}
Similarly to the proof of Lemma~\ref{lem:f1},
\[ f_2 = \sum_{j=2}^r \lambda_j^{|V(H)|-1} (a_j)_1. \]
Note that $|V(H)| - 1 = N+1$ where, recall, $N = |V(G)|$. Now compute
\begin{align*}
\EE[f_2^2] &= \sum_{j=2}^r \, \sum_{j' = 2}^r \lambda_j^{N+1} \lambda_{j'}^{N+1} \, \EE[(a_j)_1 (a_{j'})_1] \\
&= \sum_{j=2}^r \lambda_j^{2(N+1)} \\
&\le (r-1) \cdot|\lambda_2|^{2(N+1)},
\end{align*}
completing the proof.
\end{proof}

\subsection{Noise Term}

We now handle the term $f_3$. This section is devoted to proving the following.

\begin{lemma}\label{lem:f3}
Under the assumptions of Theorem~\ref{thm:main-upper}, $\EE[f_3^2] \le 4k^{k-1}D^{52k} \frac{r}{n^{k-1-\One_{k\text{ even}}}}$.
\end{lemma}

\noindent To compute $\EE[f_3^2]$, it will help to introduce an auxiliary graph $\tH$ defined as follows. Start with two disjoint copies of $H$, called $H$ and $H'$. Delete the vertices $\circ$ and $\circ'$ and connect the two leftover half-edges to form the edge $(u,u')$. This completes the description of $\tH$.

The vertices of $\tH$ can be partitioned as $V(\tH) = \{u,u'\} \sqcup V \sqcup V'$ where $V$ comes from the copy of $G$ in $H$, and $V'$ from $H'$. Similarly, the edges of $\tH$ can be partitioned as $E(\tH) = \{(u,u')\} \sqcup E \sqcup E'$ where $E$ comes from $H$, and $E'$ from $H'$.

\begin{definition}\label{def:edge-labeling-2}
Define an \emph{edge-labeling of $\tH$} to be a function $\phi: E(\tH) \to [n]$ such that $\phi(u,u') = 1$, no other edge has the label 1, no two edges in $E$ (defined above) have the same label, and no two edges in $E'$ have the same label. Let $\tPhi$ denote the set of all edge-labelings of $\tH$.
\end{definition}

\begin{definition}\label{def:vertex-labeling-2}
Define a \emph{vertex-labeling of $\tH$} to be a function $\psi: V(\tH) \to [r]$ such that $\psi$ takes at least two different values within $V \cup \{u\}$ (defined above), and $\psi$ takes at least two different values within $V' \cup \{u'\}$. Let $\tPsi$ denote the set of all vertex-labelings of $\tH$.
\end{definition}

\noindent For $\phi \in \tPhi$, $\psi \in \tPsi$, and $v \in V(\tH)$, define $T_v(\phi,\psi)$ similarly to above: letting $e_1,\ldots,e_m$ be the edges incident to $v$, and $j := \psi(v)$, let $T_v(\phi,\psi) := \lambda_j \prod_{i=1}^m (a_j)_{\phi(e_i)}$.

With the above definitions in hand, and recalling that $\Psi_3$ is the set of vertex-labelings of $H$ that take at least two different values, we can write
\begin{equation}\label{eq:f3-2}
f_3^2 = \frac{1}{|\Phi|^2} \sum_{\phi \in \tPhi} \, \sum_{\psi \in \tPsi} \, \prod_{v \in V(\tH)} T_v(\phi,\psi).
\end{equation}
Only certain ``valid'' pairs $(\phi,\psi)$ yield a term with nonzero expectation.

\begin{definition}\label{def:valid}
For $\phi \in \tPhi$ and $\psi \in \tPsi$, we say $(\phi,\psi)$ is \emph{valid} if the following holds: for each $i \in [n]$ and $j \in [r]$, there is an even number of edges $e \in E(\tH)$ with the property that $\phi(e) = i$ and exactly one endpoint of $e$ has vertex-label $j$. In fact, this even number must be 2 because, by Definition~\ref{def:edge-labeling-2}, only two edges can share the same label $i$.
\end{definition}

Note that valid pairs $(\phi,\psi)$ are precisely those for which the corresponding term in~\eqref{eq:f3-2} has an even number of factors of each $(a_j)_i$. We can now write
\begin{align}
\EE[f_3^2] &= \frac{1}{|\Phi|^2} \sum_{(\phi,\psi) \text{ valid}} \EE \prod_{v \in V(\tH)} T_v(\phi,\psi) \nonumber \\
&= \frac{1}{|\Phi|^2} \sum_{(\phi,\psi) \text{ valid}} \, \prod_{v \in V(\tH)} \lambda_{\psi(v)}.
\label{eq:f3-formula}
\end{align}

\begin{definition}
For $\psi \in \tPsi$, a \emph{region} is the preimage under $\psi$ of some $j \in [r]$. In other words, a region consists of all vertices of $\tH$ that have a particular label.
\end{definition}

For a valid $(\phi,\psi)$ pair, let $R$ denote the number of non-empty regions. By Definition~\ref{def:vertex-labeling-2} we must have $R \ge 2$. Since $(u,u')$ is the only edge with label $1$ (Definition~\ref{def:edge-labeling-2}) and $(\phi,\psi)$ is valid, $u$ and $u'$ must belong to the same region; call this region 1, and number the other non-empty regions $2,\ldots,R$. For $1 \le i \le R$, let $s_i$ denote the number of vertices in $V$ that belong to region $i$ and let $s'_i$ denote the number of vertices in $V'$ that belong to region $i$. Let $\bs_i = \min\{s_i, N - s_i\}$ and $\bs'_i = \min\{s'_i, N - s'_i\}$ where, recall, $N = |V| = |V'|$. For $1 \le i \le R$, let $\ell_i$ denote the number of edges in $E$ that ``cross'' region $i$ (i.e., have exactly one endpoint in region $i$). Since the edges of $\tH$ crossing region $i$ must be paired up with each pair having the same edge-label (Definition~\ref{def:valid}), and edge-labels cannot repeat within $E$ or $E'$ (Definition~\ref{def:edge-labeling-2}), $\ell_i$ must also be equal to the number of edges in $E'$ that cross region $i$. The total number of cross-edges (i.e., edges of $\tH$ whose endpoints have different vertex-labels) is $\ell = \sum_{i=1}^R \ell_i$. Note that $(u,u')$ is never a cross-edge since both its endpoints belong to region 1. As a consequence of the above discussion, every non-empty region must include at least one vertex from both $V \cup \{u\}$ and $V' \cup \{u'\}$.

\begin{lemma}\label{lem:ell1-bound}
For any valid $(\phi,\psi)$ pair and any $i \in \{1,2\}$,
\[ \ell_i \ge \max\{k'-1, c \bs_i, c \bs'_i\} \]
where $c > 0$ is the constant from Proposition~\ref{prop:expander}.
\end{lemma}

\begin{proof}
Recall that $u$ belongs to region 1 by convention. Let $S \subseteq V$ denote the vertices in $V$ that belong to region $i$. The case $S = \emptyset$ is possible only if $i = 1$, in which case we have $\ell_i = k'-1$. The case $S = V$ is possible only if $i = R = 2$ (since there must be at least 2 regions, each containing a vertex from both $V \cup \{u\}$ and $V' \cup \{u'\}$), in which case again $\ell_i = k'-1$. This leaves the case $0 < |S| < N$. Proposition~\ref{prop:expander} (applied to either $S$ or $V \setminus S$, whichever is smaller) tells us that the number of \emph{original} edges of $G$ (recall some edges were deleted to form $H$) crossing $S$ is at least the maximum of $k$ and $c \cdot \min\{|S|, N - |S|\} = c \bs_i$. For each edge $(v_1,v_2)$ that was deleted from $G$ to form $H$, if $(v_1,v_2)$ crosses $S$ then one of the two new edges $(u,v_1)$ or $(u,v_2)$ must cross region $i$. We conclude that at least $\max\{k,c\bs_i\}$ edges in $E$ cross region $i$. The same argument applied to $V'$ gives the bound $\max\{k,c\bs'_i\}$.
\end{proof}

\begin{lemma}\label{lem:elli-bound}
For any valid $(\phi,\psi)$ pair and any $3 \le i \le R$,
\[ \ell_i \ge \max\{k, c \bs_i, c \bs'_i\} \]
where $c > 0$ is the constant from Proposition~\ref{prop:expander}.
\end{lemma}

\begin{proof}
The proof is the same as that of Lemma~\ref{lem:ell1-bound} except now the cases $S = \emptyset$ and $S = V$ are impossible.
\end{proof}

\begin{proof}[Proof of Lemma~\ref{lem:f3}]
Combining Lemmas~\ref{lem:ell1-bound} and~\ref{lem:elli-bound}, we have for any valid $(\phi,\psi)$, the total number of cross-edges is
\begin{align}
\ell = \sum_{i=1}^R \ell_i &\ge \sum_{i=1}^R \max\{k,c\bs_i,c\bs'_i\} - 2(1+\One_{k\text{ even}}) \nonumber\\
&= \sum_{i=1}^R (k + \Delta_i) - 2(1+\One_{k\text{ even}}) \nonumber\\
&= Rk - 2(1+\One_{k\text{ even}}) + \Delta
\label{eq:ell-bound}
\end{align}
where
\begin{equation}\label{eq:Del}
\Delta_i := \max\{0, c\bs_i - k, c\bs'_i - k\}
\end{equation}
and
\[ \Delta := \sum_{i=1}^R \Delta_i. \]

We now work towards bounding~\eqref{eq:f3-formula}. Since every non-empty region must include at least one vertex from both $V \cup \{u\}$ and $V' \cup \{u'\}$, the possible values for $R$ are $2 \le R \le N+1$. The number of ways to choose the values $s_1,\ldots,s_R$ and $s_1',\ldots,s_R'$ is at most $N^{2R}$ because $0 \le s_1, s_1' \le N-1$ and for $i \ge 2$, $1 \le s_i, s_i' \le N$. Once these values are chosen, the number of ways to partition $V(\tH)$ into $R$ non-empty regions of the prescribed sizes is at most
\[ \prod_{i=1}^R \binom{N}{s_i} \binom{N}{s_i'} \le \prod_{i=1}^R N^{\bs_i + \bs'_i} \le N^{2Rk/c + 2\Delta/c}, \]
where the last step uses~\eqref{eq:Del} to conclude $\bs_i,\bs_i' \le (k+\Delta_i)/c$.

Now that the regions are chosen, we next count the number of ways to assign vertex-labels $\psi$ that respect these regions. At the same time, we will also bound the term $\prod_{v \in V(\tH)} \lambda_{\psi(v)}$ appearing in~\eqref{eq:f3-formula}. Recall that all vertices in a given region must have the same vertex-label. We consider two cases. First suppose every region contains at most $N+1$ vertices (half the total number in $\tH$). There are at most $r^R$ ways to assign the vertex-labels and, since at most half the vertices have label 1, we have $\prod_{v \in V(\tH)} \lambda_{\psi(v)} \le |\lambda_2|^{N+1}$. Now consider the other case where some ``large'' region has more than $N+1$ vertices. If we choose to assign vertex-label 1 to the large region, then there are at most $r^{R-1}$ ways to assign the remaining labels and $\prod_{v \in V(\tH)} \lambda_{\psi(v)} \le 1$; otherwise, there are at most $r^R$ ways to assign the labels and $\prod_{v \in V(\tH)} \lambda_{\psi(v)} \le |\lambda_2|^{N+1}$.

Now we count the number of ways to assign edge-labels $\phi$. Recall that the edge $(u,u')$ is required to have edge-label $1$, and no other edge can have edge-label $1$. Recall that there are $\ell/2$ cross-edges in $E$ and $\ell/2$ cross-edges in $E'$. These need to be paired up, with each cross-edge in $E$ having the same edge-label as some cross-edge in $E'$. There are $(\ell/2)!$ ways to choose the pairing and then $(n-1)^{\underline{\ell/2}}$ ways to assign edge-labels to the cross-edges, recalling the falling factorial notation~\eqref{eq:ff}. There are $|E| - \ell/2$ edges in $E$ remaining, which can have any edge-labels subject to not repeating within $E$, so there are $(n - \ell/2 - 1)^{\underline{|E| - \ell/2}}$ ways to label these edges and the same number of ways to label the rest of $E'$. (Here we have assumed $\ell/2 \le n-1$, which will indeed be the case: $\ell/2 \le |E|$ by definition, and we will see $|E| \le n/2$ below.)

Note that $|E| = \frac{kN}{2} + \frac{k'-1}{2} \le \frac{1}{2} k(N+1) = \frac{1}{2} kD$, and by the assumptions of Theorem~\ref{thm:main-upper},
\begin{equation}\label{eq:D-upper}
D \le k n^{1/52} \log n + 2.
\end{equation}
Thus, for sufficiently large $n_0$ we have $|E| \le n/2$.

Putting it all together,~\eqref{eq:f3-formula} becomes
\begin{align*}
\EE[f_3^2] &= \frac{1}{|\Phi|^2} \sum_{(\phi,\psi) \text{ valid}} \, \prod_{v \in V(\tH)} \lambda_{\psi(v)} \\
&\le \frac{1}{|\Phi|^2} \sum_{R=2}^{N+1} N^{2R + 2Rk/c + 2\Delta/c} (r^{R-1} + r^R |\lambda_2|^{N+1}) \sup_\ell \, (\ell/2)! (n-1)^{\underline{\ell/2}} \left[(n-\ell/2-1)^{\underline{|E|-\ell/2}}\right]^2.
\end{align*}
We will bound pieces of this expression separately. First, since $|\Phi| = (n-1)^{\underline{|E|}}$, we have
\begin{align*}
\frac{1}{|\Phi|^2} (\ell/2)! (n-1)^{\underline{\ell/2}} \left[(n-\ell/2-1)^{\underline{|E|-\ell/2}}\right]^2
&= (\ell/2)! (n-1)^{\underline{\ell/2}} \left[\frac{(n-\ell/2-1)^{\underline{|E|-\ell/2}}}{(n-1)^{\underline{|E|}}}\right]^2 \\
&= (\ell/2)! (n-1)^{\underline{\ell/2}} \left[\frac{1}{(n-1)^{\underline{\ell/2}}}\right]^2 \\
&\le (\ell/2)^{\ell/2} (n-\ell/2)^{-\ell/2} \\
&\le \left(\frac{\ell/2}{n-\ell/2}\right)^{\ell/2}
\intertext{and recalling $\ell/2 \le |E| \le n/2$ from above,}
&\le \left(\frac{\ell/2}{n/2}\right)^{\ell/2}.
\intertext{Recalling from~\eqref{eq:ell-bound} that $\ell \ge Rk - 2(1+\One_{k\text{ even}}) + \Delta$, this becomes}
&\le \left(\frac{\ell}{n}\right)^{\frac{1}{2}(Rk - 2(1+\One_{k\text{ even}}) + \Delta)}
\intertext{and recalling $\ell \le 2|E| \le kD$,}
&\le \left(\frac{kD}{n}\right)^{\frac{1}{2}(Rk - 2(1+\One_{k\text{ even}}) + \Delta)}.
\end{align*}

We now show $r |\lambda_2|^{N+1} \le 1$. Using $1-|\lambda_2| \le \log(1/|\lambda_2|)$ and the definition of $D$ (see Theorem~\ref{thm:main-upper}),
\begin{equation}\label{eq:lambda-to-D}
|\lambda_2|^{N+1} = |\lambda_2|^D \le |\lambda_2|^{\frac{k \log n}{1-|\lambda_2|}} = \exp\left(-\frac{k \log n}{1-|\lambda_2|} \log\frac{1}{|\lambda_2|}\right) \le \exp\left(-k \log n\right) = n^{-k}.
\end{equation}
Since~\eqref{eq:r-assum-easy} implies $r \le n^{k/2}$, this gives $r|\lambda_2|^{N+1} \le 1$ as desired, implying
\[ r^{R-1} + r^R |\lambda_2|^{N+1} = r^{R-1}(1 + r |\lambda_2|^{N+1}) \le 2r^{R-1}. \]

Combining the above,
\begin{align*}
\EE[f_3^2] &\le 2 \sum_{R=2}^{N+1} \sup_\Delta
N^{2R+2Rk/c+2\Delta/c}r^{R-1} \left(\frac{k(N+1)}{n}\right)^{\frac{1}{2}(Rk - 2(1+\One_{k\text{ even}}) + \Delta)} \\
&= \frac{2}{r}\left(\frac{n}{k(N+1)}\right)^{1+\One_{k\text{ even}}} \sum_{R=2}^{N+1} \left(k^{k/2} N^{2+2k/c} (N+1)^{k/2} \frac{r}{n^{k/2}}\right)^R \sup_\Delta \left(N^{2/c} \sqrt{\frac{k(N+1)}{n}}\right)^\Delta.
\intertext{Recall $c \ge 0.08$ (Proposition~\ref{prop:expander}), which gives $1/c \le 12.5$. Using~\eqref{eq:D-upper},
\[ N^{2/c} \sqrt{\frac{k(N+1)}{n}} \le D^{2/c+1/2} \sqrt{\frac{k}{n}} \le D^{25.5} \sqrt{\frac{k}{n}} \le \left(k n^{1/52} \log n + 2\right)^{25.5} \sqrt{\frac{k}{n}} \le 1 \]
for sufficiently large $n_0$. Since $\Delta \ge 0$, the supremum above is achieved at $\Delta = 0$ and the bound on $\EE[f_3^2]$ becomes}
&= \frac{2}{r}\left(\frac{n}{k(N+1)}\right)^{1+\One_{k\text{ even}}} \sum_{R=2}^{N+1} \left(k^{k/2} N^{2+2k/c} (N+1)^{k/2} \frac{r}{n^{k/2}}\right)^R \\
&\le \frac{2}{r}\left(\frac{n}{kD}\right)^{1+\One_{k\text{ even}}} \sum_{R=2}^{N+1} \left(k^{k/2} D^{2+2k/c+k/2} \frac{r}{n^{k/2}}\right)^R \\
&\le \frac{2}{r}\left(\frac{n}{kD}\right)^{1+\One_{k\text{ even}}} \sum_{R=2}^{\infty} \left(k^{k/2} D^{25.5k+2} \frac{r}{n^{k/2}}\right)^R
\intertext{and now~\eqref{eq:r-assum-easy} implies $k^{k/2} D^{25.5k+2} \frac{r}{n^{k/2}} \le k^{k/2} D^{27k} \frac{r}{n^{k/2}} \le 1/2$, which gives}
&\le \frac{4}{r}\left(\frac{n}{kD}\right)^{1+\One_{k\text{ even}}} \left(k^{k/2} D^{25.5k+2} \frac{r}{n^{k/2}}\right)^2 \\
&= 4 k^{k-1-\One_{k\text{ even}}} D^{51k+3-\One_{k\text{ even}}} \frac{r}{n^{k-1-\One_{k\text{ even}}}} \\
&\le 4 k^{k-1} D^{52k} \frac{r}{n^{k-1-\One_{k\text{ even}}}},
\end{align*}
completing the proof.
\end{proof}

\subsection{Putting it Together}

\begin{proof}[Proof of Theorem~\ref{thm:main-upper}]
Using Lemma~\ref{lem:f1},
\begin{align*}
\EE[(f - a_{11})^2] &= \EE[(f_1 + f_2 + f_3 - a_{11})^2] \\
&= \EE[(f_2 + f_3)^2] \\
&\le 2 (\EE[f_2^2] + \EE[f_3^2]).
\end{align*}
Recall that Lemma~\ref{lem:f2} gives $\EE[f_2^2] \le r |\lambda_2|^{2(N+1)} = r|\lambda_2|^{2D}$, and from~\eqref{eq:lambda-to-D} we have $|\lambda|^{2D} \le n^{-2k}$. Also,~\eqref{eq:r-assum-easy} implies $r \le n^{k/2}$, which gives $\EE[f_2^2] \le n^{-k}$. Combining this with Lemma~\ref{lem:f3} yields
\[ \EE[(f - a_{11})^2] \le 2(n^{-k} + 4k^{k-1}D^{52k} \frac{r}{n^{k-1-\One_{k\text{ even}}}}) \le 10 k^{k-1}D^{52k} \frac{r}{n^{k-1-\One_{k\text{ even}}}}, \]
completing the proof.
\end{proof}

\section*{Acknowledgments}

The author is indebted to Jonathan Niles-Weed, Tselil Schramm, and Jerry Li for numerous detailed discussions about this problem. The author also thanks Jingqiu Ding, Bruce Hajek, Tim Kunisky, Cris Moore, Aaron Potechin, Bobby Shi, and anonymous reviewers, for helpful discussions and comments.

\bibliographystyle{alpha}
\bibliography{main}

\end{document}